\documentclass[%
 aip,
 amsmath,amssymb,
reprint, onecolumn %%%%% Comment onecolumn to override one column format %%%%%%%
]{revtex4-2}

\usepackage{graphicx}% Include figure files
\usepackage{dcolumn}% Align table columns on decimal point
\usepackage{bm}% bold math
\usepackage{todonotes}
\usepackage[utf8]{inputenc}
\usepackage[T1]{fontenc}
\usepackage{mathptmx}
\usepackage{etoolbox}
\usepackage{tikz-cd}
\usepackage{amsthm}
\usepackage{tensor}

\newcommand{\FB}{\mathcal{F\!\!B}}
\newcommand{\NB}{\mathcal{N\!\!B}}
\newcommand{\NVB}{\mathcal{N\!V\!\!B}}

\newtheorem{thm}{Theorem}
\newtheorem{lem}{Lemma}
\newtheorem{defn}{Definition}
\newtheorem{cor}{Corollary}
\newtheorem{prop}{Proposition}

\makeatletter
\def\@email#1#2{%
 \endgroup
 \patchcmd{\titleblock@produce}
  {\frontmatter@RRAPformat}
  {\frontmatter@RRAPformat{\produce@RRAP{*#1\href{mailto:#2}{#2}}}\frontmatter@RRAPformat}
  {}{}
}%
\makeatother
\begin{document}

\preprint{AIP/123-QED}

\title[Minimally-coupled Quasilinear Symmetric Hyperbolic Systems, Naturally]{Minimally-coupled Quasilinear Symmetric Hyperbolic Systems, Naturally}
% Force line breaks with \\
\author{Eleanor March}
\affiliation{
Department of Logic and Philosophy of Science, University of California, Irvine
}
 \email{eleanor.march@uci.edu; james.owen.weatherall@uci.edu}
\author{James Owen Weatherall}
\affiliation{
Department of Logic and Philosophy of Science, University of California, Irvine
}
\affiliation{Department of Mathematics, University of California, Irvine}

\date{\today}% It is always \today, today,

\begin{abstract}
We consider symmetric hyperbolicity and well-posedness in the context of natural equations, which are closed subobjects of jet bundles over natural bundles.  We show how non-natural equations can be ``naturalized'', and we discuss how a precise version of ``minimal coupling'' offers a (provably unique, up to isomorphism) method for constructing natural equations from certain equations defined on a single manifold.  We discuss how to define symmetrizers, hyperbolizations, and constraints for minimally coupled natural equations, and we show how these can be used to recover standard existence and uniqueness results, relative to any choice of background structure (a metric, for minimally coupled systems).  Additionally, we show that minimal coupling and naturality together imply that the causal cones associated with a (first order, quasilinear, minimally coupled) naturally symmetric hyperbolic natural equation must always coincide with the metric lightcones, and that a certain rank-2 tensor, associates with a natural hyperbolization, must always be divergence free relative to the associated Levi-Civita derivative operator.
\end{abstract}

\maketitle

\section{\label{sec:intro} Introduction}

Let $\mathcal{M}_n$ be the category of smooth $n$-manifolds, with smooth embeddings as arrows; and let $\FB$ be the category of smooth fiber bundles, with smooth bundle morphisms as arrows.  A \emph{natural bundle} over $n$-manifolds is a functor $X:\mathcal{M}_n\rightarrow \FB$ such that for each $M\in\mathrm{ob}(\mathcal{M}_n)$, $XM$ is a smooth fiber bundle with base space $M$, and for each $\varphi\in\mathrm{mor}(\mathcal{M}_n)$, $X\varphi$ is a smooth bundle morphism covering $\varphi$ that is such that the maps it induces on fibers are diffeomorphisms.\citep{Palais+Terng,Kolar+etal}  (We use set-theoretic notation because it is compact, though we do not intend any commitment on whether the categories under consideration here are small or locally small.)  A ($k$th-order) \emph{natural equation} on a natural bundle $X$ is a closed subobject of $J^kX$, i.e., a pair $(E,e)$, up to isomorphism,\footnote{That is, ismorphism in the slice category of natural bundles over $J^kX$.  See \S \ref{sec:prelim}.} where $E:\mathcal{M}_n\rightarrow \FB$ is a natural bundle and $e:E\rightarrow J^k X$ is a monic natural transformation whose image, at each $M\in\mathrm{ob}(\mathcal{M}_n)$, is a closed submanifold.\citep{MarchWeatherallnaturaltheories} (Here, $J^kX$ is the $k$th jet bundle of $X$.)  One can think of $(E,e)$ as determining, for each $M\in\mathrm{ob}(\mathcal{M}_n)$, a (well-behaved) closed submanifold of the $k$th jet bundle of $XM$, and thus a certain $k$th order system of partial differential equations on sections of $XM$, which is ``natural'' (or ``generally covariant'') in the sense of being functorial over $n$-manifolds.\citep{MarchWeatherallcovariance}

Natural equations are of considerable interest in physics.  Examples of natural equations include Einstein's equation on $n$-manifolds and, for $n>4$, Einstein-Gauss-Bonnet gravity; as discussed in \S \ref{Mat+MC}, many other physically important systems, such as Maxwell's equations, can be ``naturalized'', i.e., reformulated as natural equations. For this reason and others, it is striking that (sufficiently rich) natural equations never have a well-posed initial value problem, at least in the standard sense of admitting unique solutions for given initial data.\citep{MarchWeatherallnaturaltheories}  This is true even for naturalized versions of systems such as Maxwell's equations which are symmetric hyperbolic, and would otherwise be expected to have well-posed initial value problems.  The reason is closely related to Einstein's famous ``hole argument'': since natural equations are required to be functorial over diffeomorphisms of their base space, initial data can only ever determine solutions ``up to diffeomorphism''.\footnote{Compare with the discussion by \citet{Halvorson+etal} about determinism for mathematically formulated theories.  As they argue, uniqueness up to unique isomorphism is likely enough.}  Thus, the strongest sense of well-posedness one can expect for a natural equation is that solutions exist and are unique up to (unique) diffeomorphism.  This means, in particular, that no (sufficiently rich) first-order quasilinear natural equation can be symmetric hyperbolic, since any symmetric hyperbolic system has a well-posed initial value problem.\citep{Friedrichs,Geroch1996PDEs}  And yet, as we will show, natural equations may satisfy something very close to symmetric hyperbolicity.  Moreover, certain aspects of the analysis of symmetric hyperbolic systems, such as the definition of (maximal) hyperbolizations and the integrability of constraints, take on a particularly elegant form when treated at the present level of generality, for systems of this kind.

Our goal in the present work is to describe how naturality, symmetric hyperbolicity, and well-posedness interact with one another in the special case of ``naturalized'' systems of symmetric hyperbolic equations that depend only on a Lorentzian metric to at most first order---i.e., to systems that are \emph{minimally coupled}, in a sense that we make precise below.  One of the most striking and important features of symmetric hyperbolicity is that algebraic conditions that can be checked at each point of a manifold turn out to be (almost) sufficient conditions for well-posedness.  More generally, for systems with constraints, algebraic conditions on first derivatives are sufficient for well-posedness.  As we will show, in the context of natural equations, these conditions themselves can be ``naturalized'', by defining certain tensor fields that are invariant under transformations that preserve certain specified background structure.  This procedure makes precise the idea that some minimally-coupled symmetric hyperbolic systems are ``intrinisically'' so, in the sense that they are symmetric hyperbolic without introducing any further structure beyond what is needed to define the equations themselves.  This construction will then be used to clarify the status of other conditions widely assumed for relativistic field theories---in particular, the dominant energy condition and covariant conservation of the stress--energy tensor. We go on to discuss constraints, including a novel discussion of integrability that becomes apparent in the natural equations context, and then we show how the machinery we develop implies existence and uniqueness, relative to choices of background structure.

\section{\label{sec:prelim} Natural Constructions}

We begin with preliminaries.  We have already introduced natural bundles. A zeroth-order \emph{natural operator} is a natural transformation between natural bundles, $\alpha:X\rightarrow Y$.  The category of natural bundles $\NB_{n}$ has natural bundles as objects and (zeroth-order) natural operators as arrows.\footnote{Standard presentations of this subject matter do not consider the category of natural bundles, or other categorical constructions we consider in this section.\cite{Kolar+etal}  But for reasons that we hope will come clear, we find it convenient to do so.}  This category has a terminal object $1_n$, which is the natural bundle $1_n:\mathcal{M}_n\rightarrow \FB$ that takes every $n$-manifold to itself with the identity as projection.  It also has finite products, which are fibered products in the category of fiber bundles: that is, $X\times Y$ is the natural bundle that takes $M$ to $XM\times_M YM\rightarrow M$.  As in any category, a (natural) subobject $Z$ of an object $X$ is a monomorphism $f:Z\rightarrow X$, defined up to isomorphism (in the slice category over $X$). The dual notion to subobject is that of a (natural) quotient object $Y$, which is an epimorphism $g:X\rightarrow Y$, likewise defined up to isomorphism (in the coslice category under $X$).  We will refer to subobjects and quotient objects of an object $X$ as pairs $(Z,f)$, where the second element is a representative of the relevant equivalence class and the first element is the domain or codomain of $f$, respectively (also up to isomorphism).

We are interested in natural constructions involving natural bundles. First, we introduce the category $\widehat{\NB_n}$ of functors from $\NB_n$ to itself, with natural transformations as arrows.  This category has several salient objects.  First, it has the identity functor, $id_{\NB_n}$, taking every natural bundle to itself.  (We use ``$id_{\NB_n}$'' rather than ``$1_{\NB_n}$'' because this functor is not terminal in $\widehat{\NB_n}$, and we will be considering the slice category over $id_{\NB_n}$ below.)  Second, it has an object $\hat{\pi}$, which takes every natural bundle $X$ to $1_n$, and takes every arrow to the (unique) arrow from $1_n$ to itself.  There exists a natural transformation $\pi:id_{\NB_n}\rightarrow \hat{\pi}$ whose factor at each natural bundle $X$ is the unique arrow $\pi_X:X\rightarrow 1_n$, the terminal object in $\NB_n$.  This is the natural opeation assigning bundles to their projections. Note that $\hat{\pi}$ is terminal in $\widehat{\NB_n}$.

We can now define a \emph{natural construction} on natural bundles to be an object in the slice category $\mathcal{NC}_n=\widehat{\NB_n}/id_{\NB_n}$, i.e., a pair consisting of a functor $C:\NB_n\rightarrow\NB_n$ and a natural transformation $\pi_C:C\rightarrow id_{\NB_n}$.  (At times, we suppress the natural transformation to $id_{\NB_n}$ when it is obvious from context.)  Natural constructions are constructions that act naturally (functorially) on (all) natural bundles to build new bundles over natural bundles. Many familiar examples are of this form: for instance, $T$, the functor that takes every bundle to its tangent bundle, is a natural construction; as are $T^*$, the functor that takes every bundle to its cotangent bundle; $F$, the frame bundle functor; and $J^k$, the jet functors. Natural constructions can be thought of as generalizations of natural bundles, in the following sense: there is a faithful functor from natural constructions back to $\NB_n$, defined by precomposition with $\hat{\pi}$.  That is, given any natural construction $(C,\pi_C)$, the functor $C\circ\hat{\pi}$ is a natural bundle, i.e., a functor from $n$-manifolds to fiber bundles.  (Not all natural bundles arise in this way, so the functor is not essentially surjective. It is also not full, since some natural constructions, e.g., the jet functors, and the vertical tangent functor, make non-trivial use of the full bundle structure of a natural bundle.)

Epic natural transformations between natural constructions encode the projections associated with natural constructions: that is, generically, natural constructions will produce new natural bundles over $n$-manifolds, but these will also be bundles over the total spaces of the bundle at each manifold (or other bundles).  For instance, there are natural transformation $\pi^{(k,\ell)}:J^k\rightarrow J^{\ell}$, $k\geq \ell$ whose components are the canonical projections $J^kXM\xrightarrow{\pi^{(k,\ell)}_{XM}} J^{\ell}XM$ for every $X\in\mathrm{ob}(\NB_b)$ and $M\in\mathrm{ob}(\mathcal{M}_n)$.  Similarly, there are natural transformations taking $J^k$ into $J^1\circ\cdots\circ J^1$ ($k$ copies), embedding $k$th derivatives into the first derivatives of first derivatives, $k$ times.  Similar constructions can be used to define tangent and cotangent bundles to total spaces of natural bundles, vertical bundles, and so on.  These are also related by natural transformations. For instance, given the natural constructions $(T,\pi_T)$ and $(V,\pi_T)$, there is a natural transformation $\imath_V:V\rightarrow T$ in $\widehat{\NB_n}$ embedding the vertical space into the tangent space.

For future convenience, we now note a particularly useful class of natural constructions on natural bundles that arise via a pullback construction. Consider any $X\in\mathrm{ob}(\NB_n)$, let $\hat{X}$ be the object in $\widehat{\NB_n}$ which is the constant functor onto $X$, and consider the unique maps from $id_{\NB_n}$ and $\hat{X}$ to $\hat{\pi}$, respectively.  This diagram has a pullback, i.e., an object and pair of maps defined by the standard universal property for pullbacks.  We denote the object as $X\times$, with associated maps as in the diagram below.
\begin{center}
\begin{tikzcd}
X\times \arrow[r] \arrow[d, "\pi^*"]
& \hat{X} \arrow[d] \\
id_{\NB_n} \arrow[r,"\pi"] & \hat{\pi} \\
\end{tikzcd}
\end{center}
In this case, note that since $\hat{\pi}$ is terminal, $X\times$ is just the ordinary product $\hat{X}\times id_{\NB_n}$ in $\widehat{\NB_n}$. Thus we find that $(X\times,\pi^*)$ is a natural construction, which we call the pullback of $\hat{X}$ along $\pi$, or the product with $X$.  This presentation is abstract, and so it is helpful to illustrate how it works more concretely.  Observe that the objects in this diagram are functors from $\NB_n$ to itself, and the arrows are natural transformations between those functors.  Thus, for any $Y\in\mathrm{ob}(\NB_n)$, this diagram gives rise to a different diagram, involving the factors of these natural transformations at $Y$:
\begin{center}
\begin{tikzcd}
X\times Y \arrow[r] \arrow[d, "\pi^*_Y\pi_X"]
& X \arrow[d] \\
Y \arrow[r,"\pi_Y"] & 1_n \\
\end{tikzcd}
\end{center}
But now observe that, since $1_n$ is terminal in $\NB_n$ and $\NB_n$ has finite products, $X\times Y$ is indeed just the ordinary product, or pullback of $X$ along $\pi_Y$, in $\NB_n$. So, summarizing, we see that $(X\times,\pi^*)$ will be the natural construction that takes any natural bundle $Y$ to the (fibered) product bundle $X\times Y$, regarded as a bundle over $Y$ (at each $M\in\mathrm{ob}(\mathcal{M}_n)$). This allows us to pull natural bundles back along other natural bundles, and thus to let us treat bundles over the base manifold as living over the total space of other natural bundles.

Now consider the subcategory $\NVB_n$ of $\NB_n$ consisting of natural vector bundles over $n$-manifolds, whose arrows are natural transformations that are fiberwise linear at each component.  In this case, we have additional natural constructions acting just on this subcategory, which reflect standard methods for relating vector spaces.  For instance, there is a contravariant functor $*:\NVB_n\rightarrow \NVB_n$ that takes every natural vector bundle to a ``dual'' bundle, whose typical fibers are the dual spaces to the fibers of the domain natural bundle; so, for instance, the dual bundle to the tangent bundle $T$ is the cotangent bundle, $T^*$.  Note that the dual bundle to the bundle $1_n$, $0_n$, is an initial object for $\NVB_n$; note also that $0_n$ and $1_n$ are (uniquely) isomorphic, i.e., $\NVB_n$ has a zero object.  As with the category of vector spaces, we find $*\circ *$ is naturally isomorphic to the identity functor on $\NVB_n$.  Likewise, we can define a (covariant) bifunctor $\otimes:\NVB_n\times\NVB_n\rightarrow \NVB_n$ that takes pairs of natural bundles to a natural bundles whose fibers are tensor products of the fibers of the domain.  So, $T\otimes T$ would be the natural bundle of rank $(2,0)$ tensors over $n$-manifolds.  We can also consider symmetrized and antisymmetrized tensor products of tensor bundles with themselves, which we represent with the notation $S^k$ and $\wedge^k$, respectively, for $k$-fold products.

Given the analogies to vectors spaces, it is tempting to think of the dual to a natural vector bundle as consisting of some analog of ``linear functionals'' acting on ``elements'' of the natural vector bundle.  But while that is true at each point of each object of $\mathcal{M}_n$, there is in general no natural way to identify those elements for the entire natural bundle.  Likewise, we cannot think of natural tensors bundles as acting on the natural vector bundles out of which they are built.  Even so, some facts about vector spaces do carry over, because the zero section is always natural.  In particular, given a natural vector bundle morphism $f:V\rightarrow W$, which of course is a zeroth order natural operator between $V$ and $W$ whose factors are all linear bundle morphisms, we can define a new natural vector bundle $\ker{f}$, whose fibers at each point of each $M\in\mathrm{ob}{\mathcal{M}_n}$ is the kernel of $f_M$ at that point.  This kernel object can be defined by a universal property: it will be the unique subobject $(\ker{f},i)$ such that $f\circ i$ factors through $0_n$, and given any other subobject $(U,j)$ such that $f\circ j$ factors through $0_n$, $U$ will be a subobject of $\ker(f)$ such that all relevant diagrams commute. If $f$ happens to be an epimorphism, then $(W,f)$ will be a quotient object that we may also denote $V/\ker(f)$.  Thus, natural linear operators give rise to new natural bundles.  Now, suppose we are given a subobject $(W,f)$ of a natural vector bundle $V$. Then, applying the $*$ functor, we find a quotient object $(W^*,f^*)$ of $V^*$.  At each point of each $n$-manifold $M$, we can think of this quotient object as having fibers consisting of all covectors acting on the fibers of $VM$ at that point, restricted to $W$.  Thus, $\ker{f^*}$ will be the natural bundle whose fibers are those covectors that annihilate the fibers of $W$.  It follows that every natural vector bundle subobject defines a subobject of the dual bundle, consisting of those covectors that annihilate that subobject. We have the following short exact sequence: \[ 0_n\rightarrow \ker{f^*}\xrightarrow{i}V^*\xrightarrow{f^*} W^*\rightarrow 0_n\]  We will call $\ker{f^*}\rightarrow V^*$ the \emph{cokernel} of $f:W\rightarrow V$.

The constructions described thus far all concern vector bundles over a base $n$-manifold.  But many of the bundles we will be interested in will be vector bundles defined, naturally, over the total space of a natural bundle.  These bundles will not generally be vector bundles over $M$, and so they are not objects of $\NVB_n$.  How can we think of them?

The basic idea will be to combine the constructions given so far, and consider the subcategory $\mathcal{NVC}_n$ of $\mathcal{NC}_n$ consisting of \emph{natural vectorial constructions} on natural bundles over $n$-manifolds. (Note that $\mathcal{NVC}_n$ is \emph{not} the same as the slice category $\widehat{\NVB_n}/id_{\NVB_n}$.) Objects of this category will be pairs consisting of functors $V:\NB_n\rightarrow\NB_n$ and epic natural transformations $\pi:V\rightarrow id_{\NB_n}$, which together have the property that at every natural bundle $X$ and every $n$-manifold $M$, the factor $\pi_{VXM}:VXM\rightarrow XM$ defines a vector bundle over $XM$; and whose arrows are natural transformations which, at each component $X$ and each $M\in\mathrm{ob}(\mathcal{M}_n)$, are fiberwise linear on the bundle $VXM\overset{\pi_{XM}}{\rightarrow} XM$.  This category behaves very similarly to the category $\NVB_n$ we have already discussed.  Once again, there is a zero object---the functor $1_n\times$, or equivalently $id_{\NB_n}$----and the identity natural transformation.  There is also a functor $*$ taking vectorial constructions to their duals, and a bifunctor $\otimes$ taking pairs of vectorial constructions to their tensor product, which is another vectorial construction.  Similarly, our constructions for symmetric tensors, antisymmetric tensors, the kernel, and cokernel carry over.  In addition, products with natural vector bundles are natural vectorial constructions.

We can also consider \emph{natural affine constructions}, which are defined analogously to natural vectorial constructions. The key example here will be first jet bundles. Every natural affine construction determines a natural vectorial construction, whose fibers at each point of each manifold consist of the vector space on which the corresponding affine space at that point of that manifold is modeled. For the natural affine construction $J^1$, the corresponding natural vectorial construction will be denoted $K^1$. More generally, the $k$th jet bundle $J^k$ is a natural affine construction on the subcategory of $\NB_n$ which is the image of the $k-1$st jet functor $J^{k-1}:\FB_n\rightarrow\FB_n$; likewise, $K^k$ is the corresponding natural vectorial construction, defined on the same subcategory.

Finally, we consider special constructions.  Fix a natural bundle $X$, and consider the slice category $\NB_n/X$.  This will consists of all morphisms in $\NB_n$ with codomain $X$ (or better, pairs $(Z,f)$, where $f:Z\rightarrow X$).  This slice category will include all of the natural constructions $(B,\pi_B)$ on natural bundles, at the bundle $X$, as pairs $(BX,\pi_{BX})$.  It will also include all subobjects of $X$, and subobjects of natural constructions at $X$, since the composition of monomorphisms into those with their projections to $X$ will be arrows to $X$.  Objects of this slice category will be called natural \emph{over $X$}.  We can consider vectorial constructions over $X$ and affine constructions over $X$, with all of the general claims from above carrying over.  But we will also find natural bundles that are natural over other natural bundles, but which do not result from specializing natural bundle constructions.

\section{\label{Mat+MC} Naturalization and Minimal Coupling}

We now turn back to systems of equations, with the general remarks from the previous section in mind. We begin by introducing ``naturalization''. (Here, we closely follow \citet{MarchWeatherallnaturaltheories}.)  Consider the vacuum Maxwell's equations. Note that the bundle of two-forms $F:\mathcal{M}_n\rightarrow \FB$ is a natural bundle. Given a relativistic spacetime $(M, g_{ab})$, the source-free Maxwell equations determined by $g_{ab}$ and its associated Levi-Civita connection $\nabla$ is a (vector) subbundle $\check{E}_{Max}$ of $J^1FM\rightarrow FM\rightarrow M$.\footnote{We will use $\check{\cdot}$ to indicate bundles are \emph{not} natural, and generally reserve unadorned capital Latin letters for natural bundles and natural constructions.} But $\check{E}_{Max}$ is not preserved by arbitrary diffeomorphisms (or rather, their lifts to $J^1FM$); it is straightforward to see that if $\varphi:M\rightarrow M$ is a diffeomorphism, then $J^1F\varphi[\check{E}_{Max}]=\check{E}_{Max}$ iff $\varphi_*(g_{ab})=g_{ab}$. Hence, Maxwell's equations cannot be natural.

The example reveals that the obstacle to naturality is a familiar one, related to observations due to \citet{Kretschmann}. At issue is that Maxwell's equations are typically defined relative to a some \emph{fixed} Lorentzian metric, whereas they are ``generally covariant'' only if one considers arbitrary metrics and allows diffeomorphisms to act on the metric as well as the dynamical fields.  (This distinction will be at the core of our discussion of symmetric hyperbolicity.)  We address this issue, and thereby restore naturality, by building background structure, like the Lorentzian metric relative to which Maxwell's equations are defined, into the bundle of fields.

We now proceed to present this idea more formally.

\subsection{Naturalization}

Naturalization makes use of quotient objects in $\NB_n$.  To begin, we remark on a special case of quotient objects.  Suppose that $X=Y\times Z$ is a product in $\NB_n$, i.e., there exist natural bundles $Y$, $Z$ and natural transformations $y:X\rightarrow Y$, $z:X\rightarrow Z$ such that the following diagram is a pullback in $\NB_n$:
\begin{figure}[h]
\label{fig:fiberedproductsofnaturalbundlesareproducts}
    \begin{center}
    \begin{tikzcd} X \arrow[dd, "z"] \arrow[rr, "y"] & & Y \arrow[dd, "\pi_Y"] \\ & & \\
            Z \arrow[rr, "\pi_Z"] & & 1_n
    \end{tikzcd}
\end{center}
\end{figure}

\noindent Then $(Y,y)$ and $(Z,z)$ are both quotient objects of $X$ in $\NB_n$.  We will call quotient objects in $\NB_n$ that arise from a pullback in this way \emph{natural quotient bundles}.

Let $(Y,y)$ and $(Z,z)$ be natural quotient bundles as above.  Then for any $M\in\mathrm{ob}(\mathcal{M}_n)$ and any section $\psi:M\rightarrow YM$, we can define a new bundle $XM_{|\psi}$ over $M$, with total space $XM_{|\psi}=y_M^{-1}[\psi[M]]$, and projection map $\pi_{XM|\psi}$ which is the restriction of $\pi_{XM}$ to $XM_{|\psi}$.  Then the bundle $XM_{|\psi}$ is (canonically) isomorphic (as bundles) to $ZM$.  To see this, note that the following diagram is a pullback in $\FB$, where $i:XM_{|\psi}\rightarrow XM$ denotes the subspace embedding:
\begin{figure}[h]\label{fig:pullbackaquotientbyasection}
    \begin{center}
    \begin{tikzcd} XM_{|\psi} \arrow[dd, "(\pi_{|\psi}{,}\,\mathrm{id}_M)"] \arrow[rr, "(i{,}\,\mathrm{id}_M)"] & & XM \arrow[dd, "y_M"] \\ & & \\
            1_nM \arrow[rr, "(\psi{,}\,\mathrm{id}_M)"] & & YM
    \end{tikzcd} %\caption{A pullback diagram.}
    \end{center}
\end{figure}

\noindent We can then define a map $z_{\psi}:ZM\rightarrow XM$ by $u\mapsto (u,\psi(\pi_Z(u)))$ for all $u\in ZM$, where $\pi_Z:ZM\rightarrow M$ is the projection for $ZM\rightarrow M$, and we invoke the universal property of the pullback map to get the required isomorphism. Of course, this construction can be applied symmetrically for sections of $Z$.

Now let $\check{B}\overset{\pi}{\rightarrow} M$ be a fiber bundle and $\check{E}\subset J^k\check{B}$ a closed submanifold thereof, i.e., a $k$th-order (non-natural) equation on $\check{B}$. An $\ell$th-order, $\ell\geq k$, \emph{naturalization} of $\check{E}$ consists of the following data: a pair of natural bundles $B$ and $Y$ and an $\ell$th-order natural equation $(E,e)$ on $X=B\times Y$ such that there exists a section $\psi:M\rightarrow YM$ and a bundle isomorphism $(\Phi,\varphi):\check{B}\rightarrow XM_{|\psi}$ satisfying $j^\ell\Phi[(\pi^{(\ell,k)}_{\check{B}}){}^{-1}[\check{E}]]=EM_{|\psi}$. Here $EM_{|\psi}$ is the restriction of $EM$ to $J^\ell XM_{|\psi}$, where $J^\ell XM_{|\psi}$ is defined, analogously to $XM_{|\psi}$, as a pullback in $\FB$:
\begin{figure}[h]\label{fig:pullbackaquotientbyasection}
    \begin{center}
    \begin{tikzcd}
    J^\ell XM_{|\psi} \arrow[dd, "J^\ell(\pi_{|\psi}{,}\,\mathrm{id}_M)"] \arrow[rr, "J^\ell( i{,}\,\mathrm{id}_M)"] & & J^\ell XM \arrow[dd, "J^\ell y_M"]\\
    & & \\
            1_nM \arrow[rr, "J^\ell(\psi{,}\,\mathrm{id}_M)"] & & J^\ell YM
    \end{tikzcd} %\caption{A pullback diagram.}
    \end{center}
\end{figure}
Note that in general, $EM_{|\psi}$ will not be a pullback in $\FB_n$, since the restriction of $EM$ to $J^\ell XM_{|\psi}$ may not be a fibre bundle. However, in some special cases, it will be. We return to this below.

In other words, a naturalization of an equation $\check{E}$ on a bundle $\check{B}\rightarrow{M}$ consists of a natural equation $E$ on a natural bundle $X$, along with a natural quotient bundle $(Y,y)$ of $X$, such that $\check{E}$ arises as $EM$ relative to some section of $YM$. (If $\ell>k$, it gives rise to the $\ell$th-order equation obtained from $\check{E}$ via the pullback with respect to the projection map $\pi^{(\ell,k)}_{\check{B}}$.) It follows that for any $M\in\mathrm{ob}(\mathcal{M}_n)$, different sections $\psi$ of $YM$ give rise to different (but related) equations on $BM$ (including $E$). We can think of $Y$ as encoding (natural) background structure on which the equation $\check{E}$ depends,\footnote{This idea generalizes a proposal about background structure for equations due to Geroch to the natural bundle context \citep[\S6]{Geroch1996PDEs}.  Note that Geroch imposes additional conditions on his quotient bundles, because he is contemplating the case where equations on what we call $X$ induce non-trivial equations on $Y$, as well as on $X_{|\psi}$, and he wants those equations to be well-behaved.  For present purposes, those restrictions are unnecessary, but could be imposed if needed for a particular application.} and a naturalization as a choice of how to vary $\check{E}$ as we vary choices of the background structure, including across different $n$-manifolds, under the condition that the construction must be natural. We said above that naturalization was similar to Kretschmanization, which we can now make precise: essentially, naturalization involves identifying structure necessary to define $\check{E}$ as ``generally covariant", and then reconstruing the equation accordingly.  There are generally choices involved in constructing naturalizations, and so one should not expect equations to have unique naturalizations (though, as we will show below, these may exist in the presence of additional assumptions).

Let $\check{B}$ and $\check{E}$ be as above. We will say $\check{E}$ depends only on the background structure $Y$, to at most $\ell$th order, $\ell\geq k$, if there exists an $\ell$th-order naturalization $(E,e)$ of $\check{E}$ over $X=B\times Y$ such that, for any $M\in\mathrm{ob}(\mathcal{M}_n)$, any sections $\psi, \chi:M\rightarrow YM$, and any $p\in M$, the condition $j^{\ell}_p\psi=j^{\ell}_p\chi$ implies $EM_{|\psi}\cap \pi_{J^{\ell}XM}^{-1}[p]=EM_{|\chi}\cap \pi_{J^{\ell}XM}^{-1}[p]$. (Note that we could have equivalently considered $EM_{|\psi}\cap \pi_{J^{\ell}BM}^{-1}[p]=EM_{|\chi}\cap \pi_{J^{\ell}BM}^{-1}[p]$ in this definition.)

The condition that there exists an $\ell$th-order naturalization captures the requirement that the equation $E$ explicitly depends on only $Y$ to $\ell$th order.\footnote{Of course, it might depend on less than this---in general, one can always construct further naturalizations by building in additional fields to the bundle $Y$ and its quotients or by pulling back $(E,e)$ to a higher-order jet bundle.}  The second condition, meanwhile, says that this explicit $\ell$th-order dependence on $Y$ exhausts the order of dependence on $Y$ in the equation $\check{E}$, in the sense that for any $M\in\mathrm{ob}(\mathcal{M}_n)$, sections of $YM$ which are $\ell$-equivalent at some $p\in M$ give rise to equations which agree at $p$. To get a handle on how these conditions can come apart, consider, e.g., the natural bundle $X:\bar{\mathcal{M}}_n\rightarrow\FB$ whose fibers are manifolds of pairs $(\nabla, g_{ab})$ of torsion-free affine connections and Lorentzian metrics.  (Here $\bar{\mathcal{M}}_n$ is the category of $n$-manifolds admitting a Lorentzian metric.) The bundle $G$ of Lorentzian metrics is a natural quotient of this bundle. But in general, given a first-order natural equation $(E,e)$ on $X$ and some $M\in\mathrm{ob}(\bar{\mathcal{M}}_n)$, sections $\psi, \chi:M\rightarrow GM$ satisfying the condition $j^1_p\psi=j^1_p\chi$ at some $p\in M$ will not give rise to equations which satisfy $EM_{|\psi}\cap \pi_{J^1XM}^{-1}(p)=EM_{|\chi}\cap \pi_{J^1XM}^{-1}(p)$, where $Q$ is the natural bundle whose fibers over each $M$ are torsion-free affine connections on $M$.   (To see this, consider the equation imposing metric compatibility and Einstein's equation, or the equation imposing metric compatibility and that $R^a{}_{bcd}=\mathbf{0}$.)

Now we return to Maxwell's equations.  (We will work in four spacetime dimensions for convenience.) Let $F$ be the natural bundle of $2$-forms (over $4$ manifolds) and let $G$ be the natural bundle of Lorentzian metrics on $4$ manifolds, described above. Then $X=F\times G$ is the natural bundle whose fibers at each 4-manifold are manifolds of pairs $(F_{ab},g_{ab})$, where $F_{ab}$ is a two-form and $g_{ab}$ is a Lorentzian metric.  Then $(G,g)$ is a natural quotient of $X$ that takes pairs $(F_{ab},g_{ab})$ to metrics $g_{ab}$, and our original spacetime corresponds to a section $\psi:M\rightarrow GM$. Then we can define a first-order, natural equation $(E_{Max},e)$ on $X$ by taking the source-free Maxwell's equation at every point relative to the metric and its Levi-Civita derivative operator at that point. (This makes sense, since the Levi-Civita derivative operator is the unique first-order natural operator that takes $G$ to $Q$.) The result is a naturalization of Maxwell's equations, $\check{E}_{Max}$.  Finally, observe that for any $M\in\mathrm{ob}(\bar{\mathcal{M}}_n)$, any sections $\psi, \chi:M\rightarrow GM$, and any $p\in M$, the condition $j^1_p\psi=j^1_p\chi$ implies $E_{Max}M_{|\psi}\cap \pi_{J^1FM}^{-1}[p]=E_{Max}M_{|\chi}\cap \pi_{J^1FM}^{-1}[p]$.  Thus we conclude that Maxwell's equations depend only on the metric, to at most first-order. This observation reflects the familiar fact that if two metrics agree at a point, to first order, then any two-form satisfying Maxwell's equation at that point relative to the first metric also satisfies it relative to the second, and vice-versa.

\subsection{Minimal Coupling}

We now introduce minimal coupling.\footnote{This definition follows \citet{MarchWeatherallnaturaltheories}, though the uniqueness theorem below does not appear there.  See that paper for a defense of this definition as capturing the ``informal'' notion of minimal coupling discussed elsewhere in the literature, e.g. as a `comma goes to semicolon' rule. \citep{Misneretal,Trautman1965,Readetal2018,LinnemannReadcurvaturecoupling2021}} Once again, let $G:\bar{\mathcal{M}}_n\rightarrow\FB$ be the natural bundle of Lorentzian metrics over $n$-manifolds (that admit them). Let $(E,e)$ be a first-order natural equation on a natural bundle $X$, and suppose that $(G,g)$ is a natural quotient bundle of $X=B\times G$, i.e., a quotient that arises from a pullback as in Fig. \ref{fig:fiberedproductsofnaturalbundlesareproducts}. Then we say $(E, e)$ is \emph{minimally coupled} if, for any $M\in\mathrm{ob}(\bar{\mathcal{M}}_n)$, any sections $\psi, \chi:M\rightarrow GM$, and any $p\in M$, the condition $j^1_p\psi=j^1_p\chi$ implies $EM_{|\psi}\cap \pi_{J^1XM}^{-1}[p]=EM_{|\chi}\cap \pi_{J^1XM}^{-1}[p]$. In other words, a natural equation is minimally coupled if it depends only on a Lorentzian metric, to at most first order.\footnote{Note that this definition of minimal coupling focuses on the ``matter sector'', and does not require (or even consider) any gravitational dynamics.  Note, too, that Einstein's equation is not minimally coupled, since it depends on the metric to second order.  This leads to a slight oddity relative to some uses of the expression ``minimal coupling'', which is that the system consisting of Einstein's equation and some minimally coupled matter equations is not, as a whole, minimally coupled on our definition, since the Einstein part has greater than first-order dependence.}

Similarly, let $\check{E}\subset J^1\check{B}\rightarrow \check{B}\overset{\pi}{\rightarrow}M$ be a first-order equation on some bundle $\check{B}$. A \emph{minimal coupling} of $\check{E}$ consists of a first-order naturalization $(E, e)$ of $E$ over $X = B\times G$, such that for any $M\in\mathrm{ob}(\bar{\mathcal{M}}_n)$, any sections $\psi, \chi:M\rightarrow GM$, and any $p\in M$ the condition $j^1_p\psi=j^1_p\chi$ implies $EM_{|\psi}\cap \pi_{J^1XM}^{-1}[p]=EM_{|\chi}\cap \pi_{J^1XM}^{-1}[p]$.  That is, a minimal coupling is a naturalization that is minimally coupled.

Minimal coupling is often thought of as a `prescription' for constructing general-relativistic systems of equations from special-relativistic ones.  How does this work? Suppose that we have some natural bundle $Y:\mathcal{M}_4\rightarrow \FB$ and let $(M,\eta_{ab})$ be Minkowski spacetime.  Now consider a first-order ``special-relativistic'' equation $\check{E}$ on $YM$. (What do we mean by this?  More precisely: let $\check{E}$ be a first-order equation on $YM$ that admits a minimal coupling $(E,e)$ over $X=Y\times G$, such that $\check{E}=E_{|\eta_{ab}}$.)  Now observe that any minimal coupling of $\check{E}$ will induce a minimally-coupled natural equation, which, by construction is defined over any $M\in\mathrm{ob}(\bar{\mathcal{M}}_n)$ and any section $\psi:M\rightarrow GM$. That is, a minimal coupling will generate a general-relativistic system from a special-relativistic one.

One often finds arguments in the literature that minimal coupling is ambiguous, in the sense that a given special relativistic theory may admit many, inequivalent minimal couplings.  We will now show that that is not true for minimal couplings in the present sense: that is, if a system of equations admits a minimal coupling, then it is unique, up to a natural notion of isomorphism. % It will follow that the only `ambiguity' involved in thinking of minimal coupling as a `prescription' for constructing general-relativistic systems of equations from special-relativistic ones is in the choice of the the bundle $X$---essentially, in the global topology of $X$, which must be locally a fibered product---and so is unique insofar as one has in mind a `canonical' such choice (e.g.~when $B=YM$ for some natural bundle $Y$, the bundle $X=Y\times G$, given the bundles $Y$ and $G$).
The proof relies on the following result due to \citet{FletcherWeatherallpart1} (which we have translated into the notation used in this paper):\footnote{Actually, and as \citet{FletcherWeatherallpart1} show, something stronger is true: given curve (segments) $\gamma:I\rightarrow M$, $\gamma':I'\rightarrow M'$ whose images contain $p$, $p'$ respectively, one can choose $\varphi$, $O$, and an open neighbourhood $O'$ of $p'$ such that, in addition, $\varphi^*(\gamma'\cap O')=\gamma\cap O$, and $G\varphi(\psi')(x)=\psi(x)$, and $j^1_{x}G\varphi(\psi')=j^1_{x}\psi$ for all $x\in\gamma[I]\cap O$. We will not need this in what follows, however.}
\begin{thm}[Fletcher and Weatherall, 2023, Theorem 5]\label{thm:fletcherweatherall} Let $M, M'\in\mathrm{ob}(\bar{\mathcal{M}}_n)$. Then for any sections $\psi:M\rightarrow GM$, $\psi': M'\rightarrow GM'$, and any $p\in M$, $p'\in M'$, there exist open neighbourhoods $O,O'$ of $p,p'$ and a diffeomorphism $\varphi:O\rightarrow O'$ such that $\varphi(p)=p'$ and $j^1_{p'}(G\varphi(\psi))=j^1_{p'}\psi'$.
\end{thm}

We begin with several lemmas. The first result will be relevant to much of what comes later:
\begin{lem}\label{lem:(i)}
    Let $(E, e)$ be a first-order minimally coupled natural equation on $X=B\times G$. Then for any sections $\psi:M\rightarrow GM$, $\psi':M'\rightarrow GM'$ and any $p\in M$, $p'\in M'$, there exist open neighborhoods $O$, $O'$ of $p$, $p'$ and a diffeomorphism $\varphi:O\rightarrow O'$ such that $j^1_{p'}(G\varphi(\psi))=j^1_{p'}(\psi')$ and $J^1X\varphi[ EM_{|\psi}\cap \pi_{J^1XM}^{-1}[p]]=EM'_{|\psi'}\cap \pi_{J^1XM'}^{-1}[p']$.
\end{lem}
\begin{proof}
    By Theorem \ref{thm:fletcherweatherall}, for any $M, M'\in\mathrm{ob}(\bar{\mathcal{M}}_n)$,any sections $\psi:M\rightarrow GM$, $\psi': M'\rightarrow GM'$, and any $p\in M$, $p'\in M'$, there exist open neighborhoods $O$, $O'$ of $p$, $p'$ and a diffeomorphism $\varphi:O\rightarrow O'$ such that $j^1_{p'}(G\varphi(\psi))=j^1_{p'}(\psi')$. Moreover, since $(E,e)$ is natural, we know that $J^1X\varphi[e_M\circ EM]=e_{M'}\circ EM'$. But then $J^1X\varphi[e_M\circ EM_{|\psi}\cap \pi_{J^1XM}^{-1}[p]]=J^1X\varphi[EM_{|\psi}]\cap \pi_{J^1XM'}^{-1}[p']=EM'_{|J^1X\varphi(\psi)}\cap \pi_{J^1XM'}^{-1}[p']=EM'_{|G\varphi(\psi)}\cap \pi_{J^1XM'}^{-1}[p']=EN_{|\psi'}\cap \pi_{J^1XM'}^{-1}[p']$, where we have used that smooth bundle morphisms are fiber preserving in the first equality, that $J^1X\varphi[EM]=EM'$ in the second, that $X = Y\times G$ is a natural quotient in the third, and the final equality follows from the fact that $G\varphi(\psi)$ and $\psi'$ coincide to first order at $p'$. So $\varphi$ is a diffeomorphism satisfying the conditions of the proposition.
\end{proof}

\begin{lem}\label{lem:(ii)implies(iii)}
    Let $(Y, y)$ be a natural quotient bundle of some bundle $X$. Then for any $M\in\mathrm{ob}(\mathcal{M}_n)$ and any submanifolds $S, S'\subset J^1XM$, if for all sections $\psi: M\rightarrow YM$, $S_{|\psi}=S'_{|\psi}$, then $S=S'$.
\end{lem}
\begin{proof}
    By hypothesis, we have that $y:X\rightarrow Y$ is an epic natural transformation and that there exists a natural bundle $Z$ such that $X=Z\times Y$.  The first implies that $j^1y:J^1X\rightarrow J^1Y$ is also an epic natural transformation. It then follows that, given any $x\in J^1XM$, there exists some section $\psi:M\rightarrow YM$ such that $x\in J^1ZM_{|\psi}$. Now suppose that $x\in S$. Then there exists some $\psi:M\rightarrow YM$ such that $x\in S_{|\psi}$ and hence $x\in S'_{|\psi}$ and $x\in S'$, and conversely for $S'$. But then for any $x\in J^1XM$ we have $x\in S$ iff $x\in S'$. So $S=S'$.
\end{proof}

We then have the following uniqueness result for minimal coupling:
\begin{thm}[Uniqueness of minimal coupling]\label{prop:uniqunessMC}
    Let $\check{E}$ be a first-order equation on $\check{B}\overset{\pi}{\rightarrow}M$, where $M\in\mathrm{ob}(\bar{\mathcal{M}}_n)$. Suppose that there exist minimal couplings $(E, e)$, $(E', e')$ on a natural bundle $X=B\times G$, such that there exist a section $\chi:M\rightarrow GM$ and a bundle isomorphism $(\Phi,\varphi): B\rightarrow XM_{|\chi}$ such that the restrictions $EM_{|\chi}$, $E'M_{|\chi}$ of $EM$, $E'M$ to $\chi$ coincide with $j^1\Phi[E]$. Then:
    \begin{enumerate}
        \item[(i)] For any $M\in \mathrm{ob}(\bar{\mathcal{M}}_n)$, any section $\psi:M\rightarrow GM$, and any $p\in M$, $EM_{|\psi}\cap\pi_{J^1XM}^{-1}[p]=E'M_{|\psi}\cap\pi_{J^1XM}^{-1}[p]$.
        \item[(ii)] For any $M\in \mathrm{ob}(\bar{\mathcal{M}}_n)$ and any section $\psi:M\rightarrow GM$, $EM_{|\psi}=E'M_{|\psi}$.
        \item[(iii)] For any $M\in \mathrm{ob}(\bar{\mathcal{M}}_n)$, $e_M[EM]=e'_M[E'M]$.
        \item[(iv)] $(E,e)$ and $(E',e')$ are isomorphic in the slice category $\NB_n/J^1X$.
    \end{enumerate}
\end{thm}

%%%% FROM HERE

\begin{proof}
    Since $EM_{|\chi}=j^1\Phi[E]=E'M_{|\chi}$, (i) follows from Lemma \ref{lem:(i)}. But (i) implies (ii), and (ii) implies (iii) by Lemma \ref{lem:(ii)implies(iii)}. So we just need to show that (iii) implies (iv).   But since for every $M$, $e_M$ and $e'_M$ are injective and thus diffeomorphisms onto their images, (iii) implies that $e'_M{}^{-1}\circ e_M:EM\rightarrow E'M$ is a diffeomorphism.  So we can define a natural isomorphism $i:E\rightarrow E'$, with components $i_M = e'_M{}^{-1}\circ e_M:EM\rightarrow E'M$, such that $e=e'\circ i$, from which it follows that $i$ is an isomorphism in the slice category $\NB_n/J^1X$.
\end{proof}

\citet{MarchWeatherallnaturaltheories} argue that, among the various (in fact, continuum many) formulations of the vector potential Maxwell's equations that are usually taken to be minimally coupled according to the `comma goes to semicolon' rule , only one of these is minimally coupled in the sense considered here (and, moreover, is precisely the one which is standardly taken to be the `correct' version of the vector potential Maxwell's equations).\footnote{\citet{MarchWeatherallnaturaltheories} only discuss two such candidates explicitly, but the same arguments carry over straightforwardly to any of the others.}

Another useful feature of minimally coupled systems is that minimal coupling imposes certain regularity conditions on the denaturalizations of those systems. In particular, one can show that any denaturalization of a minimally coupled natural equation is a fiber bundle:
\begin{thm}\label{thm:itsabundle}
    Let $(E, e)$ be a minimally-coupled natural equation on $X=B\times G$. Then for any $M\in\mathrm{ob}(\mathcal{M}_n)$, and any section $\psi:M\rightarrow GM$, $EM_{|\psi}$ is a fiber bundle.
\end{thm}
\begin{proof}
    Consider the map $\pi_{J^1XM|EM_{|\psi}}:EM_{|\psi}\rightarrow M$. It is clear that $\pi_{J^1XM|EM_{|\psi}}$ is surjective. Moreover, we know by lemma \ref{lem:(i)} that for any $p,q\in M$, $EM_{|\psi}\cap \pi_{J^1XM}^{-1}[p]$ and $EM_{|\psi}\cap \pi_{J^1XM}^{-1}[q]$ are isomorphic. Therefore it suffices to show that for any $p\in M$, there exists a neighbourhood $O$ of $p$ and a diffeomorphism $\varphi: O\times [EM_{|\psi}\cap \pi_{J^1XM}^{-1}[p]]\rightarrow\pi_{J^1XM|EM_{|\psi}}^{-1}[O]$. For this, we claim that for any $p\in M$, there exists an open neighbourhood $O$ of $p$ and a family $(\varphi_q)_{q\in O}$ of diffeomorphisms $\varphi_q:O_p\rightarrow O_q$ smoothly parameterized by $q$, where $O_p,O_q$ are open neighbourhoods of $p,q$ respectively, such that $\varphi_q(p)=q$ and $j^1_q(G\varphi_q(\psi))=j^1_q\psi$. To see this, let $O$ be an open subset of a geodesically-normal neighbourhood of $p$, let $\chi$ be a smooth orthonormal frame field on $O$ (such always exist), and let $\mathrm{exp}_q:U_q\subset T_qM\rightarrow M$ denote the exponential map at $q$ for any $q\in O$. Then the map $\varphi_q=\mathrm{exp}_p\circ\chi_p\circ\chi_q^{-1}\circ\exp_q^{-1}$ is a diffeomorphism satisfying the required conditions, and by definition of the exponential map the family $(\varphi_q)_{q\in O}$ can be seen to be smoothly parameterized by $q$. Now let $\varphi: O\times [EM_{|\psi}\cap \pi_{J^1XM}^{-1}[p]]\rightarrow\pi_{J^1XM|EM_{|\psi}}^{-1}[O]$ be defined as follows: $\varphi(q,x)=J^1XM\varphi_q(q,x)$ for any $q\in O$ and any $x\in EM_{|\psi}\cap \pi_{J^1XM}^{-1}[p]$. Then using that $(E,e)$ is minimally coupled, it is clear that $\varphi$ is a diffeomorphism.
\end{proof}

\subsection{Higher order equations and order reduction}

Our definitions of natural equations and of naturalization did not depend on the order of an equation.  But minimal coupling only applied to first-order equations.  Below, we will see several other cases with a similar character, where some definitions are fully general and others only apply to first-order equations.  But in fact, as has been noted by other authors, the difference in generality is not as significant as it might appear.  This is because higher-order equations can often be reduced to first-order ones in a way that allows one to draw inferences from an analysis of the reduced equation to properties of the original equation. Usually his sort of reduction is done by explicitly introducing auxiliary fields and additional equations.\citep{Geroch1996PDEs}  Here we offer a somewhat different perspective that extends more readily to natural equations.

To motivate our approach, consider an example. The vacuum Maxwell's equations, formulated as a second-order system on the bundle of one-forms, $2\nabla_n\nabla^{[n}A^{a]}=\mathbf{0}$, are in some sense equivalent to the first-order system $\nabla_nF^{na}=\mathbf{0}$, $\nabla_{[a}F_{bc]}=\mathbf{0}$ on the bundle of two-forms. Namely, if one makes the identification $F_{ab}=2\mathrm{d}_aA_b$, then one can show that any solution of the first system induces a solution of the second system, and conversely, any solution of the second system gives rise to a (family of) solutions of the first system. In other words, an order reduction of a system of equations $\check{E}$ (of some order, say, $k$) on a bundle $\check{B}$ consists of a lower-order system of equations $\check{E}'$ (say, of order $\ell$) on a bundle $\check{B}S'$, plus a rule for transforming $k$th derivatives of $\check{B}$ fields into $\ell$th derivatives of $\check{B}'$ fields, such that $\check{E}'$ is obtained from $\check{E}$ upon application of this rule, and conversely the original system $\check{E}$ is recovered from $\check{E}'$ by reversing the procedure.

These considerations suggest the following definition:
\begin{defn}[Order reduction]\label{def:orderreduced}
    Let $\check{E}$ be a $k$th-order equation on a bundle $\check{B}\rightarrow M$. A reduction of $\check{E}$ to order $\ell$ consists of a bundle $\check{B}'$, a smooth bundle morphism $(\Phi, \varphi):J^k\check{B}\rightarrow J^\ell \check{B}'$, and an $\ell$th-order equation $\check{E}'$ on $\check{B}'$ such that $\Phi[\check{E}]=\check{E}'$.
\end{defn}
In this definition, the smooth bundle morphism $(\Phi, \varphi):J^k\check{B}\rightarrow J^\ell \check{B}'$ tells us how to think about $k$-jets of $\check{B}$ fields as $\ell$-jets of $\check{B}'$ fields, whilst the condition $\Phi[E]=E'$ captures the requirement that, under this identification, $E'$ is obtained from $E$ and conversely $E$ is recovered from $E'$ in the same way. Note that the case where two equations are isomorphic (as submanifolds of jet bundles of some order) is clearly a special case of order reduction, where $\ell=k$ and $(\Phi, \varphi)$ is the $k$th jet prolongation of a smooth bundle isomorphism between $\check{B}$ and $\check{B}'$.

One adantage of definition \ref{def:orderreduced} as a treatment of order reduction is that it can be extended to natural equations.
\begin{defn}[Natural order reduction]\label{def:naturalorderreduced}
    Let $(E, e)$ be a natural, $k$th-order equation on a natural bundle $X$. A natural reduction of $(E, e)$ to order $\ell\leq k$ is an $\ell$th-order natural equation $(E',e')$ on a natural bundle  $X'$ and a natural transformations $f:J^kX\rightarrow J^\ell X'$ such that the following diagram is a pullback (in $\NB_n$):
    \begin{figure}[h]\label{fig:itturnsoutorderreductionisalsoapullbackdiagram}
        \centering
        \begin{tikzcd} E \arrow[dd] \arrow[rr, "e"] & & J^kX \arrow[dd, "f"] \\ & & \\
                E' \arrow[rr, "e'"] & & J^\ell X'
        \end{tikzcd}
    \end{figure}
\end{defn}
In other words, a natural order reduction of a natural $k$th-order equation to order $\ell$ consists of the specification of a natural, $\ell$th-order equation and a natural transformation associating the derivatives of the original equation's fields with those of the reduced equations field, which together are such that our original equation is recovered, by a universal property, from the the reduced equation and the natural transformation. Again, note that (natural) isomorphism of two natural equations is clearly a special case of natural order reduction.

These definitions are abstract, and very general. How they work is perhaps clearest when one considers examples. First, observe that lower-order systems of equations always imply equations at higher order.  This always yields order-reductions. In our current framework, we can understand this as follows. Let $\check{E}'\subset J^\ell \check{B}\rightarrow ... \rightarrow \check{B}\rightarrow M$ be an $\ell$th-order equation on $\check{B}$, let $k> \ell$, and consider the $k$th-order equation $\check{E}$ on $\check{B}$ defined via $\check{E}=(\pi^{(k,\ell)}_{\check{B}})^{-1}[\check{E}']$. Then, clearly, $\check{E}'$ is clearly an $\ell$-th order reduction of $\check{E}$, on $\check{B}$, relative to the morphism $(\pi^{(k,\ell)}_{\check{B}},\mathrm{id}_M))$. By contrast, note that generically, this sort of morphism will not produce an order reduction. That is, if $\check{E}$ is a $k$th-order equation on $\check{B}$, and $\ell<k$ then in general $\pi^{(k,\ell)}_B[\check{E}]$ is not an order reduction of $\check{E}$, because $(\pi^{(k,\ell)}_{\check{B}})^{-1}[\pi^{(k,\ell)}_{\check{B}}[\check{E}]]$ will not coincide with $\check{E}$---or in other words, one generally loses information about an equation by simply discarding its higher derivatives.

Second, consider the idea that equations at higher order can be transformed into lower-order ones by treating lower-order derivatives as independent auxiliary variables. Let $\check{E}$ and $\check{B}$ be as above, and let $\ell<k$. In general, there is a monic natural transformation $\iota^{(\ell, k-\ell)}:J^k\rightarrow J^{\ell}J^{k-\ell}$, which induces a smooth bundle morphism $\iota^{(\ell, k-\ell)}_{\check{B}}:J^k\check{B}\rightarrow J^{\ell}J^{k-\ell}\check{B}$. Thus, given a bundle $\check{B}'\rightarrow M$ and a smooth bundle morphism $(\Phi,\mathrm{id}_M):J^{k-\ell}\check{B}\rightarrow \check{B}'$, then defining $\check{E}'=j^{\ell}\Phi\circ\iota^{(\ell, k-\ell)}_{\check{B}}$, we have that $\check{E}'$ is a reduction of $\check{E}$, relative to $j^{\ell}\Phi\circ\iota^{(\ell, k-\ell)}_{\check{B}}$. Thus, we have obtained an order reduction of $\check{E}$ by treating $(k-\ell)$th derivatives of $\check{B}$ fields as independent auxiliary variables, i.e.~as $\check{B}'$ fields. Note that this framework already captures the idea that, in treating such derivatives as independent auxiliary variables, one must generally impose `additional' equations to guarantee consistency (e.g.~that auxiliary variables which correspond to exterior derivatives must be closed differential forms). That comes about at the level of the map $\iota^{(\ell, k-\ell)}_{\check{B}}$---in particular, the fact that it is a monomorphism but not an isomorphism. One can think of the points in $J^{\ell}J^{k-\ell}\check{B}$ which lie outside of the image of $\iota^{\ell, k-\ell)}_{\check{B}}$ as precisely those points which are ruled out by these additional `consistency' equations.

To make this concrete, let us return to the Maxwell example with which we began. Let $A:\mathcal{M}_n\rightarrow \FB$ be the natural bundle of one-forms, and $F:\mathcal{M}_n\rightarrow \FB$ the natural bundle of two-forms. There is a family of natural transformations $\alpha\mathrm{d}:J^1A\rightarrow F$ given by constant multiples of the exterior derivative.\footnote{In fact, these are the only such natural transformations (\citet[\S25.4]{Kolar+etal}).}  We take $\alpha=2$. Now let $(M, g_{ab})$ be a relativistic spacetime, let $E$ be the second-order equation on $AM$ defined by $2\nabla_n\nabla^{[n}A^{a]}=\mathbf{0}$, relative to the metric and its associated Levi-Civita connection at that point, and likewise let $E'$ be the first-order equation on $FM$ which consists of all points satisfying $\nabla_nF^{na}=\mathbf{0}$, $\nabla_{[a}F_{bc]}=\mathbf{0}$, relative to the metric and its associated Levi-Civita connection at that point. Then $(E', FM, j^1(2\mathrm{d})\circ\iota^{(1,1)}_{AM})$ is a reduction of $E$ to order $1$.

\section{First-Order Quasilinear Symmetric Hyperbolic Systems and Naturality}

We now turn to definining first-order quasilinear symmetric hyperbolic systems, and discussing how to extend those definitions to the natural equation context.  As in our discussion of order reduction, we will begin by presenting variants of standard definitions, on (non-natural) bundles, and then showing how those can be extended to natural equations.  First we will consider quasilinearity, since that can be defined for higher order systems; them we will consider symmetric hyperbolicity and some affiliated concepts for first-order systems.

\subsection{Quasilinearity}

We begin by presenting a convenient definition of quasilinearity and recalling its consequences.  Let $\check{E}\subset J^k\check{B}\rightarrow J^{k-1}\check{B}\rightarrow ...\rightarrow J^1\check{B}\rightarrow \check{B}\rightarrow M$ be a $k$th-order equation on some fiber bundle $\pi_{\check{B}}:\check{B}\rightarrow M$. We say $\check{E}$ is \emph{quasilinear} if $(\pi^{(k,k-1)}_{\check{B}})_{|\check{E}}:\check{E}\rightarrow J^{k-1}\check{B}$ is an affine subbundle of $\pi^{(k,k-1)}_{\check{B}}:J^k\check{B}\rightarrow J^{k-1}\check{B}$. In other words, to say that an equation is quasilinear is just to say that if $\psi, \psi':M\rightarrow \check{B}$ are both solutions of the equation, i.e.~$j^k\psi, j^k\psi'\in \check{E}$, then we also have $j^k\psi+\alpha(j^k\psi-j^k\psi')\in \check{E}$ for any number $\alpha$.

Now let $\check{E}$ be a $k$th-order quasilinear equation on $\check{B}\overset{\pi}{\rightarrow}M$. Since $\check{E}$ is an affine subbundle of $\pi^{(k,k-1)}_{\check{B}}:J^k\check{B}\rightarrow J^{k-1}\check{B}$, it gives rise to an affine bundle morphism $\kappa$---in effect, a quotient mapping---to a vector bundle over $J^{k-1}\check{B}$ (naturally considered as an affine bundle), whose fiber, at each point in $J^{k-1}\check{B}$, is the quotient space of the fiber of $J^k\check{B}$ at that point by the fiber of $\check{E}$ at that point, i.e. the bundle %$(V\check{B}\otimes S^kT^*M)/\check{U}$
$K^k\check{B}/\check{U}$,\footnote{Note that our notation here is similar to in section \ref{sec:prelim}, but since we are not considering natural bundles yet, it is being used in a more standard way.} where $\check{U}$ is the direction bundle for $\check{E}\rightarrow J^{k-1}\check{B}$, i.e., the bundle whose fiber at each point is the kernel of the affine morphism $\kappa$ (or equivalently, for which each fiber is the vector bundle on which $\check{E}\rightarrow J^{k-1}\check{B}$ is modeled). %Recall that $V\check{B}\otimes S^kT^*M$ is the direction bundle for $\pi^{(k,k-1)}_{\check{B}}:J^k\check{B}\rightarrow J^{k-1}\check{B}$.
(We are avoiding the otherwise appropriate expression ``quotient bundle'' here, since we have employed that terminology differently above.) To summarize, letting $i$ denote the inclusion of $\check{E}$ into $J^k\check{B}$, the following diagram is fiberwise exact:

\begin{figure}[h]\label{fig:fiberwiseexact}
    \centering
    \begin{tikzcd} \check{E} \arrow[dd, "\pi^{(k,k-1)}_{\check{B}|\check{E}}"] \arrow[rr, "i"] & & J^k\check{B} \arrow[dd, "\pi^{(k,k-1)}_{\check{B}}"] \arrow[rr, "\kappa"] & & K^k\check{B}/\check{U} \arrow[dd] \\ & & {} & & \\
        J^{k-1}\check{B} \arrow[rr, "\mathrm{id}_{J^{k-1}\check{B}}"] & & J^{k-1} \check{B} \arrow[rr, "\mathrm{id}_{J^{k-1}\check{B}}"] & & J^{k-1} \check{B}
    \end{tikzcd}
\end{figure}
\noindent In fact, the existence of an affine bundle morphism $\kappa$ to some vector bundle completing such a diagram is an alternative way of characterising a quasilinear equation.

In what follows, we will call the bundle $(V\check{B}\otimes S^kT^*M)/\check{U}$ the equation space, $\kappa$ the equation map, and denote vectors in the equation space with capital Latin indices. Then given a section $\sigma:M\rightarrow \check{B}$, $\check{E}$ can be expressed as an `equation' as follows: $\kappa^A(j^k\sigma)=\mathbf{0}$. (One then recovers the more usual definition of quasilinearity as `linearity of an equation in its highest derivatives' by considering the expression of this `equation' in local fibered coordinates, see e.g.~\citet[pp. 11-12]{Khavkine2012}.)  In the above construction, we also introduced a bundle $\check{U}$, which we defined as the direction bundle for $\check{E}\rightarrow J^{k-1}\check{B}$. We will call this bundle the \emph{presystem}.

The presystem of a system of quasilinear equations gives us a useful way of thinking about sources. As just noted, if $\check{E}$ is a $k$th-order quasilinear equation on $\check{B}\overset{\pi}{\rightarrow}M$, then given a section $\sigma:M\rightarrow \check{B}$, $\check{E}$ can be expressed as $\kappa^A(j^k\sigma)=\mathbf{0}$. But now suppose that we move to a different ($k$th-order, quasilinear) equation, with the same presystem. Then we would have $\kappa^A(j^k\sigma)=j^A$, where $j^A$ is some section of $K^k\check{B}/\check{U}\rightarrow J^{k-1}\check{B}$. Thus we can think of sections of $(V\check{B}\otimes S^kT^*M)/U\rightarrow J^{k-1}\check{B}$ as giving rise to a family of `sourced' equations, relative to our original `unsourced' equation $\check{E}$. Note that this notion of `sourcedness' \textit{only} makes sense relative to a choice for the original system $\check{E}$.

%Amongst other things, this means that given a vertical vector field $v^{\alpha'}$ on $B$, and in the presence of a connection $\nabla$ on $B$, we can say what it is for $v^{\alpha'}$ to be (or not to be) a source for $E$ (relative to that connection): $v^{\alpha'}$ does not source $E$ whenever $\nabla_{a_1}...\nabla_{a_k}v^{\alpha'}\in U$ (or in other words, whenever $k^A(\nabla_{a_1}...\nabla_{a_k}v^{\alpha'})=\mathbf{0}$), and does otherwise. (If $B$ is a natural, first-order bundle, then we can also consider the case where $\nabla$ is induced by a connection on $M$.)

\subsection{Natural Quasilinearity and Minimal Coupling}

The definition of quasilinearity given above readily extends to natural equations.  Let $(E,e)$ be a $k$th-order natural equation on $X:\mathcal{M}_n\rightarrow\FB$. Then $(E, e)$ is \emph{(naturally) quasilinear} if at each $n$-manifold $M$, $e_M:EM\rightarrow J^kXM$ is such that (i) $\pi^{(k,k-1)}_{XM}\circ e_M$ is epic, and (ii), for each $p\in J^{k-1}XM$,  $(e_M){}_{|(\pi^{(k,k-1)}_{XM}){}^{-1}(p)}$ is an affine map. Here, condition (i) asserts that, for each $M\in\mathrm{ob}(\mathcal{M}_n)$, $\pi^{(k,k-)}_{XM}:EM\rightarrow J^{k-1}XM$ is a smooth fiber bundle, whilst condition (ii) asserts that this bundle is an affine subbundle of $(\pi^{(k,k-1)}_{XM})_{|EM}:J^kXM\rightarrow J^{k-1}XM$.  Note that if $(E,e)$ is natually quasilinear, then there exists a subobject $(U,u)$ of $K^kX$ such that for every $M\in\mathrm{ob}(M_n)$, $UM$ is the presystem for $EM$.  We will call $(U,u)$ the presystem for $(E,e)$.\footnote{See \citet[\S 14.10]{Kolar+etal}.}

We briefly remark on how quasilinearity interacts with minimal coupling and (de)naturalization. We have the following:
\begin{thm}
\label{thm:denaturalisedquasilinear}
    Let $(E, e)$ be a minimally-coupled equation on $X=B\times G$. Then for any $M\in\mathrm{ob}(\mathcal{M}_n)$ and any section $\psi:M\rightarrow GM$, $EM_{|\psi}$ is quasilinear.
\end{thm}
\begin{proof}
    By the definition of $EM_{|\psi}$ and theorem \ref{thm:itsabundle}, we have the following commuting diagram in $\FB_n$:

    \begin{figure}[h]\label{fig:pullbacksagain}
        \centering
        \begin{tikzcd}
        EM_{|\psi} \arrow[dd, "\pi^{(k,k-1)}_{XM|psi}\circ e_M"] \arrow[rr, "(i'{,}\, \mathrm{id}_M)"] & & EM \arrow[dd, "\pi^{(k,k-1)}_{XM}\circ e_M"]
        \\ & & \\
        J^{k-1}XM_{|\psi} \arrow[dd, "\pi_{J^{k-1}XM|\psi}"] \arrow[rr, "(i{,}\,\mathrm{id}_M)"] & & J^{k-1}XM \arrow[dd, "J^{k-1}y_M"]
        \\ & & \\
        1_nM \arrow[rr, "(\psi{,}\,\mathrm{id}_M)"] & & YM
        \end{tikzcd}
    \end{figure}

    Here, $i'$ is the inclusion of $EM_{|\psi}$ into $EM$, $i$ is the inclusion of $J^{k-1}XM_{|\psi}$ into $J^{k-1}XM$, and all other maps are defined as usual. Since the lower and outer squares are pullbacks, so is the upper one. But since $(E,e)$ is quasilinear, this immediately implies that $\pi^{(k,k-1)}_{XM}\circ e_:EM_{|\psi}\rightarrow XM_{|\psi}$ is an affine subbundle of $\pi^{(k,k-1)}_{XM|\psi}:J^1XM_{|\psi}\rightarrow XM_{|\psi}$.
\end{proof}

Using theorem \ref{thm:denaturalisedquasilinear}, one has the following obvious consequence of lemma \ref{lem:(i)} (the basic idea of which is just that if two quasilinear equations coincide over some point, then so do their presystems):
\begin{thm}\label{lem:wecanshufflethepresystemaround}
    $(E,e)$ be first-order, quasilinear, minimally coupled equation with presystem $(U,u)$. Then for any $M$, $M'\in\mathrm{ob}(\bar{\mathcal{M}}_n)$, sections $\psi:M\rightarrow GM$, $\psi':M'\rightarrow GM'$, and any $p\in M$, $p'\in M'$, there exist open neighborhoods $O$, $O'$ of $p$, $p'$ and a diffeomorphism $\varphi:O\rightarrow O'$ such that $j^1_{p'}(G(\varphi)\circ\psi)=j^1_{p'}(\psi')$ and $J^1X(\varphi)[UM_{|\psi}\cap \pi_{J^1XM}^{-1}[p]]=UM'_{|\psi'}\cap \pi_{J^1XM'}^{-1}[p']$.
\end{thm}
\begin{proof}
    By lemma \ref{lem:(i)}, for any $M$, $M'\in\mathrm{ob}(\bar{\mathcal{M}}_n)$, sections $\psi:M\rightarrow GM$, $\psi':M'\rightarrow GM'$, and any $p\in M$, $p'\in M'$, there exist open neighborhoods $O$, $O'$ of $p$, $p'$ and a diffeomorphism $\chi:O\rightarrow O'$ such that $J^1X(\varphi)[EM_{|\psi}\cap \pi_{J^1XM}^{-1}[p]]=EM'_{|\psi'}\cap \pi_{J^1XM'}^{-1}[p']$. But since the fiber of $UM_{|\psi}\cap \pi_{J^1XM}^{-1}[p]$ over each point in $\pi^{-1}_{XM_{|\psi}}[p]$ consists of the direction for $EM_{|\psi}$ at that point, the condition $J^1X(\varphi)[EM_{|\psi}\cap \pi_{J^1XM}^{-1}[p]]=EM'_{|\psi'}\cap \pi_{J^1XM'}^{-1}[p']$ implies $J^1X(\varphi)[UM_{|\psi}\cap \pi_{J^1XM}^{-1}[p]]=UM'_{|\psi'}\cap \pi_{J^1XM'}^{-1}[p']$.
\end{proof}

Note that, for any minimally-coupled quasilinear natural equation $(E,e)$ on $X=G\times B$, the presystem $U$ can be viewed as a natural vectorial construction over $J^1G$, whose fiber at each $M\in\mathrm{ob}(\mathcal{M}_n)$, and each $j^1_p\psi$ in $J^1GM$, is $UM_{|\psi}\cap \pi_{J^1XM}^{-1}[p]$.

\subsection{First-order quasilinear symmetric hyperbolic equations}\label{sec:symmetrichyerbolic}

We now consider first-order quasilinear symmetric hyperbolic systems \citep{Friedrichs,John,Geroch1996PDEs,Taylor2011PDE3,Khavkine2012} We begin with a series of definitions for non-natural equations.  For convenience in what follows, we will use lower-case Greek indices for fields tangent to the total space of a bundle and lower-case Latin indices for fields tangent to the base space.  An apostrophe next to a Greek index indicates that that index is vertical.  Apostrophes can always be added to lowered (covariant) indices and removed from raised ones, but not vice-versa.  Uppercase Latin indices are used for fields valued in other spaces, specified by context.  (These conventions follow \citet{Geroch1996PDEs}.)

\begin{defn}[Symmetrizer]
    Let $\check{E}\subset J^1\check{B}\rightarrow \check{B}\rightarrow M$ be a first-order quasilinear equation, with presystem $\check{U}\rightarrow \check{B}$. A \emph{symmetrizer} is a tensor field $h\indices{^a_{\alpha'}_{\beta'}}$ on $\check{B}$, symmetric in $\alpha'$ and $\beta'$, that annihilates every tensor $\theta\indices{^{\alpha'}_a}$ in $\check{U}$ at each $p\in B$.
\end{defn}

\begin{defn}[Causal covector and vectors]
    Let $\check{E}\subset J^1\check{B}\rightarrow \check{B}\rightarrow M$ be a first-order quasilinear equation, with presystem $\check{U}\rightarrow \check{B}$, and let $h\indices{^a_{\alpha'}_{\beta'}}$ be a symmetrizer. A covector $n_a$ at some $p\in M$ is \emph{causal} for $h\indices{^a_{\alpha'}_{\beta'}}$ at some $x\in \pi^{-1}[p]$ if $n_ah\indices{^a_{\alpha'}_{\beta'}}$ is positive definite at $x$. A vector $\xi^a$ at $p$ is \emph{causal} for $h\indices{^a_{\alpha'}_{\beta'}}$ at some $x\in \pi^{-1}[p]$ if $\xi^an_a> 0$ for every causal covector for $h\indices{^a_{\alpha'}_{\beta'}}$ at $p$.
\end{defn}

Note that the causal covectors associated with a symmetriser at some point $x\in \check{B}$ form an open convex cone at $p=\pi(x)$. The causal vectors associated with that symmetrizer form a closed convex cone at $p$. One can think of the causal covectors (for a given symmetriser) at a point as representing possible initial value surfaces, and the causal vectors (for a given symmetriser) at a point as representing possible `signal propagation' velocities off of those surfaces. (The justification for this way of thinking about causal (co)vectors at a point will become clear in what follows.)

\begin{defn}[Characteristic covector]
    Let $\check{E}\subset J^1\check{B}\rightarrow \check{B}\rightarrow M$ be a first-order quasilinear equation, with presystem $\check{U}\rightarrow \check{B}$. A covector $n_a$ at some $p\in M$ is \emph{characteristic} at some $x\in \pi^{-1}(p)$ if there exists a non-zero vertical vector $v^{\alpha'}$ at $x$ such that $n_av^{\alpha'}\in \check{U}_{p}$
\end{defn}

\begin{defn}[Hyperbolization]
    Let $\check{E}\subset J^1\check{B}\rightarrow \check{B}\rightarrow M$ be a first-order quasilinear equation, with presystem $\check{U}\rightarrow \check{B}$, and let $h\indices{^a_{\alpha'}_{\beta'}}$ be a symmetrizer. Then $h\indices{^a_{\alpha'}_{\beta'}}$ is a \emph{hyperbolization} if it admits causal covectors at each $x\in \check{B}$.
\end{defn}

\begin{defn}[Maximal hyperbolization]
    Let $\check{E}\subset J^1\check{B}\rightarrow \check{B}\rightarrow M$ be a first-order quasilinear equation, with presystem $\check{U}\rightarrow \check{B}$, and let $h\indices{^a_{\alpha'}_{\beta'}}$ be a hyperbolization. Then $h\indices{^a_{\alpha'}_{\beta'}}$ is \emph{maximal} if at each $x\in B$, the boundary of the collection of causal covectors for $h\indices{^a_{\alpha'}_{\beta'}}$ contains only characteristics.
\end{defn}

\begin{defn}[Symmetric hyperbolic system]
    Let $\check{E}\subset J^1\check{B}\rightarrow \check{B}\rightarrow M$ be a first-order quasilinear equation. Then $\check{E}$ is \emph{symmetric hyperbolic} if it admits a maximal hyperbolization.
\end{defn}

Let $\check{E}$ be a quasilinear, symmetric hyperbolic equation on a bundle $\check{B}\rightarrow M$, and fix a hyperbolization $h\indices{^a_{\alpha'}_{\beta'}}$ of $\check{E}$. By initial data for $\check{E}$ (relative to the fixed choice of hyperbolization $h\indices{^a_{\alpha'}_{\beta'}}$), we mean a codimension $1$ submanifold $S$ of $M$ together with a section $\phi:S\rightarrow \check{B}$ such that at each $p\in S$ the normal $n_a$ to $S$ at $p$ is causal for $h\indices{^a_{\alpha'}_{\beta'}}$ at $\phi(p)$. Given initial data $(S, \phi)$ for $\check{E}$, by a solution of $\check{E}$ for that initial data we mean a solution $\psi:O\rightarrow \check{B}$ of $\check{E}$ (on some $O$ containing $S$) such that $\psi$ agrees with $\phi$ on $S$. Given a solution $\psi:O\rightarrow \check{B}$, a smooth curve $\gamma:I\rightarrow M$  is causal relative to $\psi$ if, whenever $\gamma(s)\in O$, its tangent vector at $\gamma(s)$ is causal for $h\indices{^a_{\alpha'}_{\beta'}}$ at $\psi(p)$. A point $p$ is an endpoint of a smooth curve $\gamma:I\rightarrow M$ iff for any open set $O$ containing $p$, there is a parameter $s_0\in I$ such that either for all $s\geq s_0$ or all $s\leq s_0$, $\gamma(s)\in O$. Then, given initial data $(S, \phi)$, the domain of dependence $D(S)$ of $S$ is the collection of points $p\in M$ such that (i)  there exists a solution $\psi:O\rightarrow \check{B}$ of $\check{E}$ for $(S, \phi)$ with $p\in O$, and (ii) for any smooth curve $\gamma:I\rightarrow M$ whose image contains $p$, and which is causal relative to solution $\psi:O\rightarrow \check{B}$ of $\check{E}$ for $(S, \phi)$ with $p\in O$, and without endpoint in $O$, then
$\gamma(I)$ intersects $S$.

We can now state the following (classic, fundamental) uniqueness result for symmetric hyperbolic systems:
\begin{thm}\label{thm:uniqueness}
    $\check{E}$ be a quasilinear, symmetric hyperbolic equation on a bundle $\check{B}\rightarrow M$, and fix a hyperbolization $h\indices{^a_{\alpha'}_{\beta'}}$ of $\check{E}$. Let $(S, \phi)$, $(S,\phi')$ be initial data for this system. Suppose that there is some open (in the submanifold topology on $S$) $S'\subset S$ on which $\phi$ and $\phi'$ agree. Then given any solutions $\psi:O\rightarrow \check{B}$, $\psi':O'\rightarrow \check{B}$ of $\check{E}$ for $(S, \phi)$, $(S,\phi')$, $\psi$ and $\psi'$ agree on $D(S')$.
\end{thm}
(There is, of course, also a corresponding existence result.  We return to that below, because it will require a treatment of constraints.) Thus, initial data on $S$ suffice to fix solutions throughout $D(S)$. It is for this reason that one can think of causal vectors as representing possible signal propagation velocities.

In setting up this theorem, we did not mention either characteristics or maximality. So what role do they play? To answer this, note that the uniqueness result above applies to equations \textit{with a particular hyperbolization}. In general, equations admit multiple hyperbolizations, with different initial value surfaces associated with them. But we are really interested in the initial value formulation for the full equation, not the hypberbolized version. In other words, one would like to know what the possible initial value surfaces are for the full equation---i.e.~what the largest collection of causal covectors is.

This is where characteristics become important. To see this, note that characteristic covectors cannot be causal. This is because if $n_a$ is characteristic at a point, then
$n_ah\indices{^a_{\alpha'}_{\beta'}}$ cannot be positive definite there, because there exists a non-zero vector $v^{\alpha'}$ such that $n_ah\indices{^a_{\alpha'}_{\beta'}}v^{\alpha'}=\mathbf{0}$---namely, the vector $v^{\alpha'}$ such that $n_av^{\alpha'}\in U$ at that point. Thus, if the boundary of the collection of causal covectors for some hyperbolization contains only characteristics, we know there cannot be another hyperbolization with wider causal cones. Thus, the maximal hyperbolizations reflect the full causal structure of the equation with which we began.

To see how hyperbolizations work with an example, consider, again, the Maxwell system on some 4-manifold $M$, with metric $g_{ab}$.  (Since we are not working in the natural equation context, we mean here a particular system of equations $\check{E}$ on a bundle $\check{F}\rightarrow M$ of $2-$forms on $M$.) The presystem $\check{U}$ for this system has, as fibers, vector spaces of tensors $u_{abc}$, antisymmetric in $b$ and $c$, that satisfy $g^{ab}u_{abc}=\mathbf{0}$ and $u_{[abc]}=\mathbf{0}$. Fix $x\in \check{F}$. The symmetrizers will be symmetric tensors $h\indices{^a_{\alpha'}_{\beta'}}$---i.e.~tensors $h\indices{^{amnxy}}=h\indices{^{a[mn][xy]}}=h\indices{^{axymn}}$---with the property that $h\indices{^{amnxy}}u_{amn}=\mathbf{0}$ for all $u_{amn}\in \check{U}$. This can happen in two ways: either $h\indices{^a_{mnxy}}$ contracts the first two indices of $\check{U}$ with $g^{ab}$, or else it antisymmetrizes on all three indices. So we can get hyperbolizations for the Maxwell system by considering tensors of the form $T^{abmnxy}=\alpha(g^{b[m}g^{n][x}g^{y]a}+g^{b[x}g^{y][n}g^{m]a})+\beta g^{ab}g^{m[x}g^{y]n}$. (Here, the first term annihilates $U$ by antisymmetrizing on all three indices and the second term annihilates $U$ by taking its trace; $\alpha$ and $\beta$ are numbers chosen to ensure positive definiteness---which we get if $\alpha=2\beta$. We take $\alpha=1/2$ by convention.) We get a hyperbolization from $T^{abmnxy}$ by contracting with any timelike covector $t_b$, so that the possible hyperbolizations are tensors $h\indices{^{amnxy}}=T^{abmnxy}t_b=1/2(t^{[m}g^{n][x}g^{y]a}+t^{[x}g^{y][n}g^{m]a})+1/4t^ag^{m[x} g^{y]n}$. This is a symmetrizer by construction, and is a hyperbolization since $n_ah\indices{^{amnxy}}$ is positive definite for all $n_a$ timelike and co-oriented with $t_a$. (This follows from the fact that the Maxwell stress-energy tensor satisfies the DEC.) Finally, note that the characteristic covectors for this system are precisely the null covectors: for any null $n_a$, we have $n_an_{[b}x_{c]}\in \check{U}$ for any $x_a$ orthogonal to $n_a$. So these hyperbolizations are also maximal.

\subsection{Symmetric hyperbolicity and naturality}

We now reach the main question motivating the present work, which is how to think about symmetric hyperbolicity in the context of natural equations.  The first observation, which has been pointed out by others, is that natural equations cannot be symmetric hyperbolic.\citep{Geroch1996PDEs,MarchWeatherallnaturaltheories} There is an easy way to see this.\footnote{A somewhat different argument is given by \citet{Geroch1996PDEs}, though he does not explicitly discuss this in terms of naturality. Geroch's `combined system' would be a natural theory, in the present sense.} Say that a natural equation $(E, X)$ is `sufficiently rich' if it admits solutions $\psi:M\rightarrow XM$ such that there exists a diffeomorphism $\varphi:M\rightarrow M$ for which $X\varphi\circ \psi \neq \psi\circ \varphi$. Meanwhile, a theory `admits a well-posed initial value problem' if, given suitable `initial data' on an $(n-1)$-submanifold of $M$, there exists an open set $O$ containing that submanifold and a smooth solution on $O$, agreeing with the initial data on that submanifold, which is unique in the sense that any other such solution agreeing with the initial data coincides with that solution throughout $O$.  We have the following result (compare with Theorems \ref{thm:uniqueness} and \ref{thm:existence}):
\begin{thm}[Weatherall and March, 2025]
    No sufficiently rich natural theory admits a well-posed initial value problem.
\end{thm}
Now observe that quasilinear symmetric hyperbolic systems in the sense discussed thus far \emph{do} admit a well-posed initial value problem.\footnote{Again, we are glossing over subtleties here, primarily having to do with completeness and integrability of constraints for such systems, which e.g.~enter into the characterisation of `suitable' initial data---see \citet{Geroch1996PDEs}---but these details will not matter for our purposes in this section and we address them below.} Thus no sufficiently rich natural theory can be symmetric hyperbolic. (There is also a local version of the argument, which begins by noting that if any covector is characteristic for a natural theory, then all covectors must be characteristic.  This can be proved using techniques in section \ref{sec:causalityetc}.

It follows from the foregoing that naturalized equations cannot be symmetric hyperbolic---even in cases where the equation with which one started \emph{is} symmetric hyperbolic.  The reason has to do with how diffeomorphisms act on background structure.  One would like to understand the conditions under which one can take a naturalized equation, or any equation with a suitable quotient structure as arises from naturalization, and recover a symmetric hyperbolic system, relative to some (or any) choice of background structure. The Maxwell example we have been considering is a motivating case.  Maxwell's equations can be expressed relative to any Lorentz-signature metric; and relative to any Lorentz-signature metric, they are symmetric hyperbolic.  Naturalizing Maxwell's equations captures the sense in which they depend on the choice of background structure.  But can we capture the way in which that dependence interacts with symmetric hyperbolicity, so that the fact that the symmetric hyperbolicity also depends on the metric is apparent?

%We can sharpen the question as follows.  Let $(E, e)$ be a natural, quasilinear equation on a natural bundle $X$, and $(Y, y)$ be a natural quotient of $X$.  Fix $M\in\mathrm{ob}(\mathcal{M}_n)$ and some section $\psi:M\rightarrow \mathfrak{Y}M$ such that $EM_{|\psi}$ is symmetric hyperbolic. When does this also occur for every other choice of background structure (on every $M\in\mathrm{ob}(\mathcal{M}_n)$)?\footnote{An alternative approach would be to develop a theory of `almost symmetric hyperbolicity' which is appropriate for natural equations, and show that whenever a natural equation is `almost symmetric hyperbolic' in this sense, it gives rise to symmetric hyperbolic equations relative to any (appropriate) choice of background structure. This approach would take us too far afield, here, but we suggest that it would be a worthwhile one to pursue. For suggestions in this direction, see \citet{Geroch1996PDEs}.}

In what follows, we will focus on the case of minimally coupled equations.  Before we discuss a sufficient condition for an affirmative answer to the question just posed, we will first highlight a strategy that does not work, because it is informative about the structure of the problem. Suppose that we have a minimally coupled quasilinear equation $(E,e)$ on $X=G\times B$. Suppose further that we could find a first-order natural operator $H:G\rightarrow T\otimes S^2V^*B$ with the property that $H(M, \psi)$ is a (maximal) hyperbolization of $EM_{|\psi}$. (One could think of such an operator as an assignment of a hyperbolization to every choice of background structure.) Then we would have that $H(M, \psi)$ is a (maximal) hyperbolization of $EM_{|\psi}$ for every $M\in\mathrm{ob}(\bar{\mathcal{M}}_n)$ and every section $\psi:M\rightarrow GM$. This follows immediately from the definition of a maximal hyperbolization, the definition of a natural operator, and theorems \ref{thm:fletcherweatherall}, \ref{thm:denaturalisedquasilinear}, and \ref{lem:wecanshufflethepresystemaround}.

However, this strategy will not be successful, because the only first-order natural operator $H:G\rightarrow T\otimes S^2V^*B$ is the zero operator (\citet[\S 33.2]{Kolar+etal}). However, this approach does suggest an alternative strategy, which is to find a first-order natural operator acting on $G$ which returns not maximal hyperbolizations, but mixed index tensor fields that encode a space of maximal hyperbolizations, relative to some further choices. Indeed, as we saw above, this was exactly how we defined maximal hyperbolizations for the Maxwell example---by introducing a tensor $T\indices{^a_{\alpha'}^b_{\beta'}}$, depending only on the metric to first order, with the property that for any timelike $t_a$, $t_nT\indices{^n_{\alpha'}^b_{\beta'}}$ is a hyperbolization.  (It follows that non-zero operators of that character at least sometimes exist.) This motivates the following definitions:

\begin{defn}[Natural symmetrizer]
    Let $(E,e)$ be a minimally-coupled quasilinear natural equation on $X=G\times B$ with presystem $(U,u)$. A natural symmetrizer is a natural transformation $H:G\rightarrow S^2T\otimes S^2V^*B$, such that $U$ is a subobject of the kernel of its induced morphism $H:(T^*\otimes VB)\times G\rightarrow (T\otimes V^*B)\times G$ of (special) natural vectorial constructions over $G$.
\end{defn}

Note that we have defined natural symmetrizers as natural operators on $G$.  This means they do not depend on field values, only on metrics.  One could also consider a more general definition that permits field dependence, though that would complicate arguments in the following.

We say that a natural symmetrizer is a \emph{natural maximal hyperbolization} iff there exist some $M\in\mathrm{ob}(M_n)$, some section $\psi:M\rightarrow GM$, and some covector $\xi_a$ such that $H(M,\psi)=T\indices{^a_{\alpha'}^b_{\beta'}}$ is such that $\xi_nT\indices{^n_{\alpha'}^a_{\beta'}}$ is a maximal hyperbolization of $EM_{|\psi}$.

At a first pass, this definition of a natural maximal hyperbolization may appear to be insufficiently general and/or natural. But in fact, we could have proceeded differently. The following theorem provides an equivalent characterization of these natural maximal hyperbolizations:
\begin{thm}\label{thm:equivalence}
    Let $(E,e)$ be a minimally-coupled quasilinear natural equation on $X=G\times B$ and let $H:G\rightarrow S^2T\otimes S^2V^*B$ be a natural symmetrizer. The following are equivalent:
    \begin{itemize}
        \item[i.] $H$ is a natural maximal hyperbolization.
        \item[ii.] For any $M\in\mathrm{ob}(M_n)$ and any section $\psi:M\rightarrow GM$, every null covector is characteristic for $EM_{|\psi}$, and the induced morphism $H: (S^2T^*)\times G\rightarrow (S^2_+V^*B)\times G$ of special natural vectorial constructions over $G$ restricts to a morphism $H:(I^*_\pm\times I^*_\pm)\times G\rightarrow (S^2_+V^*B)\times G$ or $H:(I^*_\pm\times I^*_\mp)\times G\rightarrow (S^2_+V^*B)\times G$ of natural constructions over $G$.
    \end{itemize}
    Here, $I^*_\pm\times I^*_\pm$, $I^*_\pm\times I^*_\mp$ are the special natural constructions over $G$ whose fibers are manifolds of pairs of co-directed (respectively oppositely-directed) timelike covectors, and $S^2_+V^*B$ is the natural bundle of positive definite vertical tensors on $B$.
\end{thm}

We begin with several lemmas.

\begin{lem}\label{lem:isn'titsymmetric}
    Let $(E,e)$ be a minimally-coupled quasilinear natural equation on $X=G\times B$. Then for any $M\in\mathrm{ob}(\mathcal{M}_n)$, any section $\psi:M\rightarrow GM$, and any $x\in XM_{|\psi}$, if any timelike (spacelike, non-zero null) covector at $p=\pi_{XM_{|\psi}}[x]$ is characteristic for $EM_{|\psi}$ at $x$, then every timelike (spacelike, null) covector at $p$ is characteristic for $EM_{|\psi}$ at some $y\in \pi^{-1}_{XM_{|\psi}}[p]$.
\end{lem}
\begin{proof}
    Recall that a covector $n_a$ at a point $p\in M$ is characteristic for $EM_{|\psi}$ at a point $x\in\pi_{XM_{|\psi}}^{-1}[p]$ iff there exists a non-zero vertical vector $v^{\alpha'}$ at $x$ such that $n_av^{\alpha'}$ is in $UM_{|\psi}$ at $x$, where $(U,u)$ is the presystem for $(E,e)$. We know, by theorem \ref{lem:wecanshufflethepresystemaround}, that for any $M$, $M'\in\mathrm{ob}(\mathcal{M}_n)$, sections $\psi:M\rightarrow GM$, $\psi':M'\rightarrow GM'$ and any $p\in M$, $p'\in M'$, there exist open neighborhoods $V$, $V'$ of $p$, $p'$ and a diffeomorphism $\chi:V\rightarrow V'$ such that $J^1X(\chi)[UM_{|\psi}\cap \pi_{J^1XM}^{-1}[p]]=UM_{|\psi'}\cap \pi_{J^1XM'}^{-1}[p']$, and hence that for any such diffeomorphism $\chi$, $n_av^{\alpha'}\in UM_{|\psi}$ at $x$ iff $K^1X(\chi)(n_av^{\alpha'})\in UM'_{|\psi'}$ at $X(\chi)(x)$. Now consider the case where $M\simeq\mathbb{R}^n$ and $\psi$ is the Minkowski metric on $M$. Then given any two unit timelike (spacelike, non-zero null) covectors $n_a$, $n'_a$ at $p$, there exists a diffeomorphism $\chi:M\rightarrow M$ such that $G(\chi)(\psi)=\psi$ and $\chi_*n_a=n'_a$. In particular, and using that the elements of $U$ at each point form a vector space and that the zero covector is always characteristic, this implies that for any $x\in\pi^{-1}_{XM_{|\psi}}[p]$ if any timelike (spacelike, non-zero null) covector with respect to $\psi$ is characteristic for $EM_{|\psi}$ at $x$, then every timelike (spacelike, null) covector at $p$ with respect to $\psi$ is characteristic for $EM_{|\psi}$ at some $y\in\pi^{-1}_{XM_{|\psi}}[p]$. But then theorems \ref{thm:fletcherweatherall} and \ref{lem:wecanshufflethepresystemaround} and lemma \ref{lem:(i)} imply that for any other $M'$, $\psi'$, $p'$ the same is true.
\end{proof}

\begin{lem}\label{lem:characteristicandcausal}
    Let $(E,e)$ be a minimally-coupled quasilinear natural equation on $X=G\times B$, and let $H:G\rightarrow S^2T\otimes S^2V^*B$ be a natural maximal hyperbolization thereof. Let $M\in\mathrm{ob}(\mathcal{M}_n)$ and $\psi:M\rightarrow GM$ be such that there exists some $\xi_a$ such that $\xi_nT\indices{^n_{\alpha'}^b_{\beta'}}$ is a maximal hyperbolization of $EM_{|\psi}$, where $T\indices{^a_{\alpha'}^b_{\beta'}}=H(M,\psi)$. Then for any $x\in XM_{|\psi}$, if some timelike (null, spacelike) covector at $p=\pi_{XM_{|\psi}}[x]$ is causal for $\xi_nT\indices{^n_{\alpha'}^b_{\beta'}}$ at $x$, then no timelike (non-zero null, spacelike) covector is characteristic for $EM_{|\psi}$ at any $y\in\pi^{-1}_{XM_{|\psi}}[p]$.
\end{lem}
\begin{proof}
    Suppose otherwise, i.e., there is some timelike (null, spacelike) $\eta_a$ at $p$ which is causal for $\xi_nT\indices{^n_{\alpha'}^b_{\beta'}}$ at $x\in\pi^{-1}_{XM_{|\psi}}[p]$, but there is some timelike (non-zero null, spacelike) $\eta'_a$ at $p$ which is characteristic for $EM_{|\psi}$ at some $y\in\pi^{-1}_{XM_{|\psi}}[p]$. Then by lemma \ref{lem:isn'titsymmetric}, $\eta_a$ is characteristic for $EM_{|\psi}$ at some $z\in\pi^{-1}_{XM_{|\psi}}[p]$. But by the definition of a natural maximal hyperbolization, we know that for any $p\in M$ and any $x,z\in\pi^{-1}_{XM_{|\psi}}[p]$, a covector at $p$ is causal for $\xi_nT\indices{^n_{\alpha'}^b_{\beta'}}$ at $x$ iff it is causal for $\xi_nT\indices{^n_{\alpha'}^b_{\beta'}}$ at $z$. So $\eta_a$ is causal for $\xi_nT\indices{^n_{\alpha'}^b_{\beta'}}$ at $z$. This is impossible, since characteristics cannot be causal.
\end{proof}

\begin{lem}\label{lem:nocausalspacelike}
    Let $(E,e)$ be a minimally-coupled quasilinear natural equation on $X=G\times B$, and let $H:G\rightarrow S^2T\otimes S^2V^*B$ be a natural maximal hyperbolization thereof. Let $M\in\mathrm{ob}(\mathcal{M}_n)$ and $\psi:M\rightarrow GM$ be such that there exists some $\xi_a$ such that $\xi_nT\indices{^n_{\alpha'}^b_{\beta'}}$ is a maximal hyperbolization of $EM_{|\psi}$, where $T\indices{^a_{\alpha'}^b_{\beta'}}=H(M,\psi)$. Then for any $x\in XM_{|\psi}$ no spacelike covectors at $p=\pi_{XM_{|\psi}}[x]$ are causal for $\xi_nT\indices{^n_{\alpha'}^b_{\beta'}}$.
\end{lem}
\begin{proof}
    Suppose otherwise, i.e., there is some $x\in XM_{|\psi}$ and some spacelike $\eta_a$ at $p=\pi_{XM_{|\psi}}[x]$ which is causal for $\xi_nT\indices{^n_{\alpha'}^b_{\beta'}}$. By lemma \ref{lem:characteristicandcausal}, no spacelike vector is characteristic for $EM_{|\psi}$ at $x$. Since $\xi_nT\indices{^n_{\alpha'}^b_{\beta'}}$ is maximal, the boundary of the set of causal covectors for $\xi_nT\indices{^n_{\alpha'}^b_{\beta'}}$ at $x$ must contain only characteristics. So every spacelike covector at $p$ must be causal for $\xi_nT\indices{^n_{\alpha'}^b_{\beta'}}$ at $x$. But this is impossible, since if $\eta_a$ is spacelike and $\eta_m\xi_nT\indices{^n_{\alpha'}^m_{\beta'}}$ is positive definite, then $-\eta_a$ is spacelike and $-\eta_m\xi_nT\indices{^n_{\alpha'}^m_{\beta'}}$ is negative definite.
\end{proof}

\begin{lem}\label{lem:timelikecausal}
    Let $(E,e)$ be a minimally-coupled quasilinear natural equation on $X=G\times B$, and let $H:G\rightarrow S^2T\otimes S^2V^*B$ be a natural maximal hyperbolization thereof. Let $M\in\mathrm{ob}(\mathcal{M}_n)$ and $\psi:M\rightarrow GM$ be such that there exists some $\xi_a$ such that $\xi_nT\indices{^n_{\alpha'}^b_{\beta'}}$ is a maximal hyperbolization of $EM_{|\psi}$, where $T\indices{^a_{\alpha'}^b_{\beta'}}=H(M,\psi)$. Then for any $x\in XM_{|\psi}$, a covector is causal for $\xi_nT\indices{^n_{\alpha'}^b_{\beta'}}$ at $x$ iff it is future-directed timelike or iff it is past-directed timelike.
\end{lem}
\begin{proof}
    By assumption, $\xi_nT\indices{^n_{\alpha'}^b_{\beta'}}$ admits some causal covector, $\eta_a$, at $x\in\pi^{-1}_{XM_{|\psi}}[p]$. Suppose $\eta_a$ is timelike. Then by lemma \ref{lem:characteristicandcausal}, no timelike covector is characteristic for $EM_{|\psi}$ at $x$. By maximality, the boundary of the set of causal covectors for $\xi_nT\indices{^n_{\alpha'}^b_{\beta'}}$ at $x$ must contains only characteristics. So every timelike covector at $p$ codirected with $\eta_a$ is causal. Moreover, no timelike covector oppositely-directed with $\eta_a$ can be causal by lemma \ref{lem:nocausalspacelike} and the fact that the set of causal covectors for $\xi_nT\indices{^n_{\alpha'}^b_{\beta'}}$ is convex.

    Conversely, suppose that some non-zero null covector $\eta_a$ is causal for $\xi_nT\indices{^n_{\alpha'}^b_{\beta'}}$ at $x$. By lemma \ref{lem:characteristicandcausal}, no non-zero null covector at $p$ is characteristic for $EM_{|\psi}$ at $x$. Since the boundary of the set of causal covectors for $\xi_nT\indices{^n_{\alpha'}^b_{\beta'}}$ at $x$ contains only characteristics, this means every null covector codirected with $\eta_a$ is causal. But neither the set of non-zero null nor the set of non-zero null and timelike covectors codirected with $\eta_a$ is open and convex.
\end{proof}
\begin{cor}\label{cor:nullcharacteristic}
    Let $(E,e)$ be a minimally-coupled quasilinear natural equation on $X=G\times B$, and let $H:G\rightarrow S^2T\otimes S^2V^*B$ be a natural maximal hyperbolization thereof. Then for any $x\in XM_{|\psi}$, every null covector at $p=\pi_{XM_{|\psi}}(x)$ is characteristic for $EM_{|\psi}$.
\end{cor}

\begin{lem}\label{lem:thefirstargument}
    Let $(E,e)$ be a minimally-coupled quasilinear natural equation on $X=G\times B$, and let $H:G\rightarrow S^2T\otimes S^2V^*B$ be a natural maximal hyperbolization thereof. Let $M\in\mathrm{ob}(\mathcal{M}_n)$ and $\psi:M\rightarrow GM$ be such that there exists some $\xi_a$ such that $\xi_nT\indices{^n_{\alpha'}^b_{\beta'}}$ is a maximal hyperbolization of $EM_{|\psi}$, where $T\indices{^a_{\alpha'}^b_{\beta'}}=H(M,\psi)$. Then for any covector $\eta_a$, $\eta_nT\indices{^n_{\alpha'}^b_{\beta'}}$ is a maximal hyperbolization of $EM_{|\psi}$ iff $\eta_a$ is timelike.
\end{lem}
\begin{proof}
    Fix some $x\in XM_{|\psi}$, and let $p=\pi_X(x)$. Since $\xi_nT\indices{^n_{\alpha'}^b_{\beta'}}$ is a maximal hyperbolization of $EM_{|\psi}$, we know that all future-directed or all past-directed timelike covectors are causal for $\xi_nT\indices{^n_{\alpha'}^b_{\beta'}}$ at $p$ by lemma \ref{lem:timelikecausal}. (Without loss of generality, we consider the future-directed case.) By the symmetry of $T\indices{^a_{\alpha'}^b_{\beta'}}$, this implies that for any future-directed unit timelike covector $t_a$, $t_nT\indices{^n_{\alpha'}^b_{\beta'}}$ is a hyperbolization of $EM_{|\psi}$ at $x$, with $\xi_a$ causal for $t_nT\indices{^n_{\alpha'}^b_{\beta'}}$ at $x$. We know that, for any other future-directed unit timelike $t'_a$ at $p$, there exists an open neighbourhood $V$ of $p$ and a diffeomorphism $\varphi:V\rightarrow V$ such that $\varphi(p)=p$, $j^1_pG\varphi(\psi)=j^1_p\psi$, and $\varphi_*t'_a=t_a$, and hence that under any such diffeomorphism, $H(V,\psi)(p)=H(V,G\varphi(\psi))(p)$ and $J^1X(\varphi)[UM_{|\psi}\cap \pi_{J^1XM}^{-1}[p]]=UM_{|\psi}\cap \pi_{J^1XM}^{-1}[p]$ (This follows from the arguments used in lemma \ref{lem:isn'titsymmetric}.) It follows that for any such $\varphi$, $\varphi_*(\xi_a)$ is causal for $t_nT\indices{^n_{\alpha'}^b_{\beta'}}$ at $x$.

    $(\Rightarrow)$ By lemmas \ref{lem:nocausalspacelike} and \ref{lem:timelikecausal}, $\xi_a$ cannot be spacelike or null.

    $(\Leftarrow)$ Since $t_nT\indices{^n_{\alpha'}^b_{\beta'}}$ is a hyperbolization, what remains to be shown is maximality. Again by the above, it follows that every timelike covector codirected with $\xi_a$ is causal for $t_nT\indices{^n_{\alpha'}^b_{\beta'}}$ at $x$, and hence that $t_nT\indices{^n_{\alpha'}^b_{\beta'}}$ is maximal by corollary \ref{cor:nullcharacteristic}. The same argument suffices to establish the result for $t_a$ past-directed, substituting $\xi_a$ for $-\xi_a$.
\end{proof}

\begin{proof}[Proof of theorem \ref{thm:equivalence}]
    $(ii\Rightarrow i)$ By definition, for any $M\in\mathrm{ob}(\mathcal{M}_n)$, any section $\psi:M\rightarrow GM$, and any timelike $t_a$, $t_nT\indices{^n_{\alpha'}^b_{\beta'}}$ is a maximal hyperbolization of $EM_{|\psi}$ whose causal cones consist of all future-directed or all past-directed timelike covectors at each $x\in M$.

    $(i\Rightarrow ii)$ This follows immediately from the definition of a natural symmetrizer and lemmas\ref{lem:timelikecausal} and \ref{lem:thefirstargument}, making use of the argument of lemma \ref{lem:isn'titsymmetric} to establish that every null covector is characteristic for $EM_{|\psi}$, for any $M\in\mathrm{ob}(M_n)$ and section $\psi:M\rightarrow GM$.
\end{proof}

It immediately follows from theorem \ref{thm:equivalence} that if $(E,e)$ is a minimally-coupled quasilinear natural equation which admits a natural maximal hyperbolization, then every denaturalization of $(E,e)$ (on any $M\in\mathrm{ob}(\mathcal{M}_n)$) is symmetric hyperbolic. We now turn to a (partial) converse result. Given a minimally-coupled quasilinear natural equation $(E,e)$ that is symmetric hyperbolic relative to some denaturalization, then one can isolate necessary and sufficient conditions for $(E,e)$ to admit a natural maximal hyperbolization, at least in the case where $B$ is a first-order natural vector or affine bundle:

\begin{thm}
    Let $(E,e)$ be a minimally-coupled quasilinear natural equation on $X=G\times B$, where $B$ is a first-order vector or affine bundle, and suppose that, for some $M\in\mathrm{ob}(\mathcal{M}_n)$ and some section $\psi:M\rightarrow GM$, $EM_{|\psi}$ is symmetric hyperbolic. Then $(E,e)$ admits a natural maximal hyperbolization iff there exists a maximal hyperbolization $h\indices{^a_{\alpha'\beta'}}$ of $EM_{|\psi}$ such that at any $x\in XM_{|\psi}$ there exists a timelike covector $t_a$ at $p=\pi_X(x)$ such that for any open neighbourhood $V$ of $p$ and any local diffeomorphism $\varphi:V\rightarrow V$ such that $\varphi(p)=p$ and $j^1_p\psi=j^1_pG\varphi(\psi)$, $(\varphi_*(t_a))h\indices{^a_{\alpha'\beta'}}=t_a(T\otimes S^2V^*B)(\varphi)(h\indices{^a_{\alpha'\beta'}})$.
\end{thm}
\begin{proof}
    $(\Rightarrow)$ It suffices to observe that if $H$ is a natural maximal hyperbolization of $(E,e)$, and writing $H(M,\psi)=T\indices{^a_{\alpha'}^b_{\beta'}}$, then for any timelike $t_a$, $t_nT\indices{^n_{\alpha'}^b_{\beta'}}$ and $t_a$ are a maximal hyperbolization covector pair satisfying the conditions of the theorem.

    $(\Leftarrow)$ Suppose $h\indices{^a_{\alpha'\beta'}}$, $t_a$ are a maximal hyperbolization and timelike covector satisfying the conditions of the theorem. (Without loss of generality we may assume $t_a$ is unit.) Fix some open neighbourhood $V$ of $p$ diffeomorphic to $\mathbb{R}^n$. Consider the tensor $T\indices{^a_{\alpha'}^b_{\beta'}}$ at $x$ defined as follows: for any unit timelike $\xi_a$ at $p$ codirected with $t_a$, $\xi_nT\indices{^n_{\alpha'}^b_{\beta'}}=(T\otimes S^2V^*B)(\varphi_\xi) (h\indices{^b_{\alpha'\beta'}})$, where $\varphi_\xi$ is any local diffeomorphism $\varphi_\xi:V\rightarrow V$ such that $\varphi_\xi(p)=p$, $j^1_p\psi=j^1_pG\varphi_\xi(\psi)$, and $\varphi_{\xi*}t_a=\xi_a$. (This is well-defined, since any two such local diffeomorphisms agree on their lifts via $T\otimes S^2V^*B$ at $p$, since $B$ is a first-order vector or affine bundle.) By construction, and making use of lemma \ref{lem:wecanshufflethepresystemaround}, $T\indices{^a_{\alpha'}^b_{\beta'}}$ is symmetric in its vertical indices, annihilates $UM_{|\psi}$ in its second pair of indices. Moreover, $T\indices{^a_{\alpha'}^b_{\beta'}}$ is symmetric in its upper two indices, since for any pair of unit timelike $\xi_a$, $\eta_a$ at $p$, $\xi_n\eta_mT\indices{^n_{\alpha'}^m_{\beta'}}=\eta_m(T\otimes S^2V^*B)(\varphi_\xi) (h\indices{^m_{\alpha'\beta'}})=(\varphi_{\eta*}(t_m))(T\otimes S^2V^*B)(\varphi_\xi) (h\indices{^m_{\alpha'\beta'}})=S^2V^*B(\varphi_\eta)(t_m(T\otimes S^2V^*B)(\varphi_\eta^{-1}\circ\varphi_\xi) (h\indices{^m_{\alpha'\beta'}}))=S^2V^*B\varphi_\eta(((\varphi_\eta^{-1})_*\xi_m)h\indices{^m_{\alpha'\beta'}})=\xi_m(T\otimes S^2V^*B)(\varphi_\eta)(h\indices{^m_{\alpha'\beta'}})=\xi_m\eta_nT\indices{^n_{\alpha'}^m_{\beta'}}$. Moreover, it is straightforward to see that $T\indices{^n_{\alpha'}^b_{\beta'}}$ has a unique extension to a natural operator $H:G\rightarrow S^2T\otimes S^2V^*B$ (this follows from theorem \ref{thm:fletcherweatherall}).
\end{proof}

In summary: whilst a natural, minimally coupled, quasilinear system of equations cannot be symmetric hyperbolic, we can make sense of what it is for such a system of equations to be symmetric hyperbolic relative to any choice of background structure---i.e., relative to any Lorentzian metric---and we can, at least in some cases, do so in a natural way, by defining a natural operator that gives rise to a maximal hyperbolization at every manifold, for every choice of metric. Sufficient for this is the existence of a natural hyperbolization $H:G\rightarrow S^2(T\otimes V^*B)$.  We will call natural equations for which this construction works -- that is, first-order quasilinear natural equations that are minimally coupled and admit a natural hyperbolization -- (first-order, quasilinear, minimally-coupled) \emph{naturally symmetric hyperbolic}.  The equations themselves are not symmetric hyperbolic, but their symmetric hyperbolicity depends, naturally, on the same background structure that the equations do.

\section{Causality, energy conditions, and conservation}\label{sec:causalityetc}

Before proceeding to discuss constraints and the initial value problem, we pause to state and prove several results that show the value of this way of thinking about minimally-coupled naturally symmetric hyperbolic systems.  In brief, naturality requires hyperbolizations to be defined ``uniformly'' across manifolds, in a way that forces one to define them using resources that are available at each manifold---which in this case requires us to build them out of the metric and structures determined by it, such as volume forms.  This, in turn, allows us to prove general results regarding the causal structure and energy-momentum properties of minimally-coupled natural symmetric hyperbolic systems. In particular, we show that natural maximal hyperbolizations satisfy the dominant energy and conservation conditions.  This provides a link between natural hyperbolizations and nice properties of stress--energy tensors in general relativity.  This link is useful because in many cases of interest in physics, the Hilbert stress-energy tensor happens to be a maximal natural hyperbolization; we thus see that any matter whose stress-energy tensor suffices to secure symmetric hyperbolicity will also satisfy these other conditions.

%\begin{lem}\label{lem:nocharacteristictimelike}
%     Let $(E,e)$ be a first-order quasilinear minimally-coupled natural equation on a natural bundle $X$. Then for any $M\in\mathrm{ob}(\bar{\mathcal{M}}_n)$, any section $\psi:M\rightarrow GM$, and any $p\in M$, no timelike covector at $p$ is characteristic for $EM_{|\psi}$ at $p$.
%\end{lem}
%\begin{proof}
%    Suppose otherwise, i.e.~there exists some timelike $n_a$ at $p$ which is characteristic for $EM_{|\psi}$ at $p$. By lemma \ref{lem:isn'titsymmetric}, this implies that all timelike $n_a$ are characteristic for $e_M(EM)_\psi$. By lemma \ref{lem:nullcharacteristic}, we know that all null covectors are characteristic for $e_M(EM)_\psi$. In particular, this implies that no spacelike covectors can be characteristic, since if any spacelike covector is characteristic then by lemma \ref{lem:isn'titsymmetric} all spacelike vectors must be, but if all covectors are characteristic then $e_M(EM)_\psi$ cannot be symmetric hyperbolic. It follows that any maximal hyperbolization $h\indices{^a_{\alpha'}_{\beta'}}$ of $e_M(EM)_\psi$ must have, as its collection of causal covectors, some open convex cone of spacelike covectors bounded by null covectors at $p$. But no such set of covectors exists.
%\end{proof}

\begin{thm}\label{conj:DEC}
   Let $(E,e)$ be a first-order quasilinear minimally-coupled naturally symmetric hyperbolic natural equation on a natural bundle $X$. Let $H:G\rightarrow S^2(T\otimes V^*Y)$ be a maximal natural hyperbolization for $(E,e)$. Then for any $M\in\mathrm{ob}(\mathcal{M}_n)$, any section $\psi:M\rightarrow GM$, any $x\in XM$, and any $v^{\alpha'}$ at $x$, either $T\indices{^a_{\mu'}^b_{\nu'}}v^{\mu'}v^{\nu'}$ satisfies the dominant energy condition or $-T\indices{^a_{\mu'}^b_{\nu'}}v^{\mu'}v^{\nu'}$ does, where $T\indices{^a_{\mu'}^b_{\nu'}}=H(M,\psi)$.
\end{thm}
\begin{proof}
    By theorem \ref{thm:equivalence}, for any timelike $t_a$ at $p=\pi_X(x)$, $n_a$ is causal for $t_aT\indices{^a_{\mu'}^b_{\nu'}}$ iff it is timelike and either co-oriented with $t_a$ or oppositely oriented.  Assume $n_a$ is causal iff it is co-oriented with $t_a$.  (Otherwise, consider $-T\indices{^a_{\mu'}^b_{\nu'}}$.) Then, by the definition of a maximal hyperbolization, if $t_a$, $\xi_a$ are any two cooriented timelike covectors, then if $T\indices{^a_{\alpha'}^b_{\beta'}}$ is a tensor on $XM$ satisfying the conditions of the proposition, $t_n\xi_mT\indices{^n_{\alpha'}^m_{\beta'}}$ is positive definite, and hence for all $v^{\alpha'}$, either $v^{\alpha'}\neq\mathbf{0}$ and $t_n\xi_mT\indices{^n_{\mu'}^m_{\nu'}}v^{\mu'}v^{\nu'}>0$, or $v^{\alpha'}=\mathbf{0}$ and $t_n\xi_mT\indices{^n_{\mu'}^m_{\nu'}}v^{\mu'}v^{\nu'}=0$. But the latter case, of course, obtains only if $T\indices{^a_{\mu'}^b_{\nu'}}v^{\mu'}v^{\nu'}=\mathbf{0}$.
\end{proof}

We now turn to the conservation condition, understood as a condition on natural hyperbolizations.  Before stating the results, some further preliminaries are needed.  First, observe that for any first-order natural bundle $Y$ and any $M\in\mathrm{ob}(\mathcal{M}_n)$, a derivative operator (linear connection) $\nabla$ on $M$ gives rise to a connection $\omega^{\alpha'}{}_{\beta}$ on $YM$.  This is because a linear connection on $M$ gives rise to a principal connection on the frame bundle to $M$, and $YM$ can be seen as associated to that frame bundle.
In particular, since every Lorentzian metric gives rise (naturally) to a Levi-Civita derivative operator, we find a first-order natural operator $\nabla:G\rightarrow VY\otimes T^*Y$ taking metrics to connections on any natural bundle $Y$.

We now state our main result, which characterizes the divergence of a natural hyperbolization relative to this Levi-Civita derivative operator.
    \begin{thm}\label{conj:CC}
         Let $(E,e)$ be a first-order quasilinear minimally-coupled naturally symmetric hyperbolic natural equation on a natural bundle $X = G\times Y$, where $Y$ is first order. Let $H:G\rightarrow S^2(T\otimes V^*Y)$ be a maximal natural hyperbolization for $(E,e)$. Then for any $M\in\mathrm{ob}(\mathcal{M}_n)$, any section $M\rightarrow GM$, and any smooth vertical vector field $v^{\alpha'}$ on $M$, $\nabla_n(T\indices{^n_{\mu'}^a_{\nu'}}v^{\mu'}v^{\nu'})=2T\indices{^n_{\mu'}^a_{\nu'}}v^{\nu'}j\indices{^{\mu'}_n}$, where $T\indices{^a_{\mu'}^b_{\nu'}}=H(M,\psi)$ and $j\indices{^{\alpha'}_a}$ is any representative of $[\nabla_av^{\alpha'}]_{/U}\in VXM/U$
    \end{thm}
\begin{proof}
    We know by the symmetry of $T$ that $\nabla_n(T\indices{^n_{\mu'}^a_{\nu'}}v^{\mu'}v^{\nu'})=v^{\mu'}v^{\nu'}\nabla_nT\indices{^n_{\mu'}^a_{\nu'}}+2T\indices{^n_{\mu'}^a_{\nu'}}v^{\nu'}\nabla_nv^{\mu'}$. Since $\nabla_av^{\alpha'}\in VXM$, and  $T\indices{^a_{\alpha'}^b_{\beta'}}$ annihilates every element of $U\subset VXM$, this implies that $\nabla_n(T\indices{^n_{\mu'}^a_{\nu'}}v^{\mu'}v^{\nu'})=v^{\mu'}v^{\nu'}\nabla_nT\indices{^n_{\mu'}^a_{\nu'}}+2T\indices{^n_{\mu'}^a_{\nu'}}v^{\mu'}j\indices{^{\nu'}_n}$, where $j\indices{^{\alpha'}_a}$ is any representative of $[\nabla_av^{\alpha'}]_{/U}\in VXM/U$. But now consider $\nabla_nT\indices{^n_{\alpha'}^a_{\beta'}}$. This is, in particular (the image at $M$ of) a first-order natural operator $G\rightarrow T\otimes S^2V^*Y$. But the only such natural operator is the zero operator. So $\nabla_nT\indices{^n_{\alpha'}^a_{\beta'}}=\mathbf{0}$.
\end{proof}

Note that Theorem \ref{conj:CC} does not say that a natural hyperbolization necessarily has vanishing divergence, even when evaluated at solutions.  This is to be expected: sometimes equations are sourced, and sometimes the equation itself has energy instabilities (as with the tachyon equation, $\nabla_a\nabla^a\varphi +m^2\varphi = \mathbf{0}$.  But it does allow us to state the conditions under which a natural hyperbolization \emph{is} divergence-free, as captured by the following corollary.

\begin{cor}
    Under the conditions of theorem \ref{conj:CC}, $T\indices{^a_{\mu'}^b_{\nu'}}v^{\mu'}v^{\nu'}$ satisfies the CC iff $T\indices{^n_{\mu'}^a_{\nu'}}v^{\nu'}\nabla_nv^{\mu'}=\mathbf{0}$.
\end{cor}

\section{Constraints}

A full treatment of the existence and uniqueness properties of quasilinear symmetric hyperbolic equations in this setting requires a discussion of constraints.  Constraints play two roles in the theory of symmetric hyperbolic partial differential equations, both related to the existence side of the fundamental existence and uniqueness results.  They imply conditions that initial data must satisfy in order to be compatible with solutions on any open set around that initial data.  And they also imply integrability conditions that must be satisfied by the equation as a whole to admit any solutions at all.   We begin with the following definition, inspired by \citet{Geroch1996PDEs}.
\begin{defn}[Constraints]
    Let $\check{E}\subset J^1\check{B}\rightarrow \check{B}\rightarrow M$ be a first-order quasilinear equation, with presystem $\check{U}\rightarrow \check{B}$. A \emph{constraint field} is a tensor field $c\indices{^b^a_{\alpha'}}$ on $\check{B}$, antisymmetric in $a$ and $b$, that annihilates every tensor $\theta\indices{^{\alpha'}_a}$ in $\check{U}$ at each $p\in \check{B}$.
\end{defn}

To see how constraints impose (or capture) consistency conditions on initial data, suppose we have a first-order quasilinear equation $\check{E}\subset J^1\check{B}\rightarrow \check{B}\xrightarrow{\pi} M$, with presystem $\check{U}\rightarrow \check{B}$.  Suppose $c^{ba}{}_{\alpha'}$ is a constraint field for this equation.  Fix any point $x\in \check{B}$ and any covector $n_a$ at $\pi(x)$.  Then $n_bc^{ba}{}_{\alpha'}$ is a linear functional on elements of $K^1\check{B}$ at $x$.  Let $X_{|(x,n_a)}\subseteq (K^1\check{B})_{|x}$ be intersection of the kernels of all such maps for fixed $n_a$, ranging over constraints at $x$.  (It suffices to consider a basis for the vector space of constraints at $x$. The subscripts reflect that this space depends on both $x$ and $n_a$.)  By the definition of $c^{ba}{}_{\alpha'}$, $X_{(x,n_a)}\supseteq \check{U}_{|x}$.  So what we find is that $X_{|(x,n_a)}$ is itself a presystem, for some family of equations, which contains $U$; since a larger presystem means a small number of equations, it follows that $n_bc^{ba}{}_{\alpha'}$ determines a subsystem of $E$ at $x$, associated with the covector $n_a$.  Finally, note that if we act with $n_bc^{ba}{}_{\alpha'}$ on a field $\vartheta^{\alpha'}{}_a$, components parallel to $n_a$ will vanish, due to the antisymmetry of $c^{ba}{}_{\alpha}$.  This means that for any connection $\omega^{\alpha'}{}_{\beta}$ on $\check{B}$, $n_bc^{ba}{}_{\alpha'}\omega^{\alpha'}{}_{\beta}$ takes derivatives, relative to the connection, only along the surface normal to $n_a$. .

It follows that $n_bc^{ba}{}_{\alpha}$ determines a subsystem of equations that hold only on surfaces \emph{normal} to $n_a$.  These are implied by $\check{E}$ in the sense that any solution to $\check{E}$ must also solve these equations.  If we now consider an $n-1$ dimensional surface $S$ in $M$ on which we wish to specify initial data, and consider the normal $n_a$ to that surface, the construction just described generates, for each constraint field, an implied equation on $S$, i.e., depending only on derivatives tangent to $S$, that must be satisfied if there is to exist any solution to $\check{E}$ on a neighborhood of $S$ that coincides with the initial data on $S$.  These implied equations are the constraint equations.

In order for $\check{E}$ to admit solutions for arbitrary initial data satisfying the constraint equations, there must be enough constraints to fixed all implied degrees of freedom on $S$, without there being too many to overconstrain the system.  When this happens, the constraints are said to be \emph{complete}.  Constraints are complete at a point $x$ of $\check{B}$, relative to a covector $n_a$ at $\pi(x)$, if $\dim(X_n)-\dim(U)=\dim(F)$, where $F$ is the typical fiber of $B$.  We say the constraints are simply \emph{complete} if they are complete at all points of $\check{B}$, for all non-zero covectors; likewise, they are complete for timelike covectors if the condition holds for all points of $\check{B}$ and all timelike covectors.  To see why these definitions capture the required condition, note that the dimensions of the corresponding equation spaces are $\dim(M)\times \dim(F)-\dim(X_n)$ and $\dim(M)\times \dim(F)-\dim(U)$.  Thus the expression says that the number of equations in the system, minus the number of constraints, is equal to the number of unknowns for the system.

We now turn to integrability.  Our treatment will be somewhat idiosyncratic, though, we hope, enlightening.  We are interested, here, in whether there exist smooth local solutions to $\check{E}$ through every point $x\in\check{B}$.  If there \emph{does} exist such a solution, $\varphi:O\rightarrow B$, through $x$, it follows that there also exists a field $(\nabla\varphi)^{\alpha}{}_a$, the differential of $\varphi$, on $\varphi[O]$, which (a) lies everywhere in $E$, and (b) is such that for any vector fields $\xi^a$, $\eta^a$ on $O$, $[(\nabla\varphi)^{\alpha}{}_a\xi^a,(\nabla\varphi)^{\beta}{}_b\eta^b]$ is tangent to $\varphi[O]$.  This sort of expression should be familiar: it is a kind of local flatness condition for solutions, relative to themselves (as surfaces in $\check{B}$). In the present context, it means that if the equation is ``integrable'' in the sense of admitting smooth solutions, then it is ``integrable'' in the sense of satisfying local integrability conditions in the form of local curvature conditions.  In other words, equations admitting solutions through every point must contain ``flat'' distributions of horizontal subspaces of $T\check{B}$.

With this basic insight in mind, we now introduce our integrability condition.  First, observe that for any open set $O\in\check{B}$, and any section $W^{\alpha}{}_a:O\rightarrow J^1\check{B}$, we can define a curvature tensor in the standard way as that tensor such that, for any fields $\xi^{a}$, $\eta^{a}$ on $\pi[O]$,
\[
    \Omega^{\alpha'}{}_{ab}\xi^{a}\eta^{b} = \omega^{\alpha'}{}_{\beta}[W^{\kappa}{}_{a}\xi^a,W^{\beta}{}_{b}\eta^b],
\]
where $\omega^{\alpha'}{}_{\beta}$ is the unique (surjective) projection onto the vertical space $V\check{B}$ at each point of $O$ such that $\omega^{\alpha'}{}_{\beta}W^{\beta}{}_a=\mathbf{0}$.  Then we have:
\begin{defn}
 Let $\check{E}\subset J^1\check{B}\rightarrow \check{B}\rightarrow M$ be a first-order quasilinear equation, with presystem $\check{U}\rightarrow \check{B}$. A constraint field $c\indices{^b^a_{\alpha'}}$ on $\check{B}$ is \emph{integrable} if, for any point $x\in\check{B}$, there exists a neighborhood $O\subseteq\check{B}$ of $x$ such that every section $W^{\alpha}{}_a:O\rightarrow \check{E}$ has curvature $\Omega^{\alpha'}{}_{ab}$ satisfying $c^{ab}{}_{\alpha'}\Omega^{\alpha'}{}_{ab}=\mathbf{0}$.  The equation $\check{E}$ is \emph{integrable} if all of its constraint fields are.
\end{defn}
What this definition says is that an equation is integrable if any local connection that lies within the equation is flat ``up to the direction of $E$'', i.e., its curvature vanishes when acted upon by any linear functional on curvatures that annihilates the presystem, $U$.

This way of thinking about integrability is quite different from how Geroch handles it, even though we follow his treatment closely in other respects.  To make contact, we first introduce a derivative operator $\nabla$ on $\check{B}$ with the property that the derivative of every vertical vector is vertical; and we extend this derivative operator to act on mixed index tensors by stipulating that $\nabla_{\alpha}(\nabla\pi)^{\alpha}{}_{\beta}=\mathbf{0}$.  (The derivative operator is otherwise arbitrary.)  One can then derive the following expression, for any constraint field $c^{ba}{}_{\alpha'}$ and section $W^{\alpha}{}_a:O\rightarrow \check{E}$:
\[
c^{ba}{}_{\alpha'}\Omega^{\alpha'}{}_{ab}= 2W^{[\kappa}{}_aW^{\beta]}{}_{b}(\nabla_{\kappa} c^{ba}{}_{\alpha'}\omega^{\alpha'}{}_{\beta}),
\]
where $\Omega^{\alpha'}{}_{ab}$ is the curvature associated with $W^{\alpha}{}_a$ and $\omega^{\alpha'}{}_{\beta}$ is as defined above.  Now consider the tensor field $\nabla_{[\kappa} c^{ba}{}_{|\alpha'|}\omega^{\alpha'}{}_{\beta]}$ on the right-hand side of the equation.  Note that the expression on the left-hand side will vanish iff $\nabla_{[\kappa} c^{ba}{}_{|\alpha'|}\omega^{\alpha'}{}_{\beta]}$ annihilates any pair of sections of $E$.  (Why any pair?  Because the difference between $W^{\alpha}{}_a$ and another section of $E$ will be given by a field in $U$, which will vanish when contracted with $\nabla_{[\kappa} c^{ba}{}_{|\alpha'|}\omega^{\alpha'}{}_{\beta]}$.)  It follows that $\nabla_{[\kappa} c^{ba}{}_{|\alpha'|}\omega^{\alpha'}{}_{\beta]}$, as a biaffine functional on fibers of $J^1\check{B}$ at each point of $O$, factors through the quotient map $k^A:J^1\check{B}\rightarrow J^1\check{B}/U$, taking derivatives to the ``equation space'' associated with $\check{E}$, in both arguments.  In other words, we have proved the following proposition:
\begin{prop}
    A constraint field $c^{ba}{}_{\alpha'}$ for a first-order quasilinear equation $\check{E}\subseteq J^1\check{B}\rightarrow\check{B}\rightarrow M$ with presystem $U$ is integrable iff (a) $\nabla_{[\kappa} c^{ba}{}_{\alpha']}$ annihilates pairs of elements of the fibers of the presystem at each point and (b) $\nabla_{[\kappa} c^{ba}{}_{|\alpha'|}\omega^{\alpha'}{}_{\beta]}=\sigma_A{}^b{}_{\kappa}k^{Aa}{}_{\beta} + \sigma_A{}^a{}_{\beta} k^{Ab}{}_{\kappa}$, for some tensor field $\sigma^{A}{}_a{}_{\alpha}$, with $k$, $\nabla$, and $\omega$ as in the previous paragraphs.\footnote{We have abused notation here, since $k^A$ is an affine map and we have written it as a tensor.  The action should be understood as limited to fibers of $J^1\check{B}$.}
\end{prop}
Compare this result to the analysis of integrability in \citet[\S4]{Geroch1996PDEs}.

We now turn to constraints in the natural context.  Unlike symmetrizers and hyperbolizations, constraints form a vector space at each point of $B$.  This is important for present purposes, because constraints, as particular fields, do not lend themselves to the natural theories context: we apparently cannot think of constraints as natural operators on $G$ or $B$.  But the \emph{vector space} of constraints may be natural, in the following sense.
\begin{defn}[Naturalized Constraint Space]
Let $(E,e)$ be a first-order quasilinear natural equation on a natural bundle $X\rightarrow M$, with presystem $(U,u)$.  A \emph{natural constraint bundle} associated with $E$ is a natural bundle $C$ for which the following diagram is a pullback in $\mathcal{NB}_n$:

\begin{figure}[h!]
\label{fig:fiberedproductsofnaturalbundlesareproducts}
    \begin{center}
    \begin{tikzcd} C \arrow[dd] \arrow[rr] & & (T\times)\otimes\ker(u^*) \arrow[dd, "i"] \\ & & \\
            \left((\wedge^2T\times)\otimes V^*\right)X \arrow[rr, "i"] & & \left(((\otimes^2T)\times)\otimes V^*\right)X
    \end{tikzcd}
\end{center}
\end{figure}

where $\ker({u^*})$, recall, is the cokernel of $u$. In other words, it is a (special) natural tensorial construction over $X$ whose fibers consist of tensors with the symmetries of a constraint, as defined above, that annihilate the presystem $(U,u)$ at each point of $X$ at each $n-$manifold.
\end{defn}

In what follows, we assume that the equations we are considering admit natural constraint bundles.  This is an additional regularity assumption that we are imposing.  In the presence of this assumption, no significant changes are needed to define when constraints are complete in the natural context; one simply applies the dimensional calculation to the fibers of the relevant natural bundles at each $n-$manifold.  (Of course, whether the constraints are complete will not depend on the choice of $n-$manifold.)  Integrability, by contrast, now admits a much more elegant definition.  First, let $\Omega_X:J^1\circ J^1X\rightarrow VX\otimes \wedge^2 (T\times X)$ be the natural operation taking connections on $X$ and their first derivatives to their associated curvature tensors.  (Note that we can think of $\Omega_X$ as arising from a natural transformation between natural constructions $\Omega:J^1\circ J^1\rightarrow V\otimes\wedge^2(T\times)$ whose components, at each natural bundle $X$, is the natural operation that assigns curvatures to connections on $X$.)  Now consider the natural bundle $\nabla E$ defined by following pullback diagram:
\begin{figure}[h]
\label{fig:fiberedproductsofnaturalbundlesareproducts}
    \begin{center}
    \begin{tikzcd} \nabla E \arrow[dd] \arrow[rr] & & J^1\circ J^1X \arrow[dd, "\imath"] \\ & & \\
           J^1E \arrow[rr, "i"] & & J^1\circ J^1X
    \end{tikzcd}
\end{center}
\end{figure}
Here $\nabla E$ is the natural bundle of fields, derivatives of fields, and second derivatives of fields that satisfy $E$.  Since this is a subbundle of $J^1\circ J^1X$, $\Omega_X$ acts on $\nabla E$ by restriction, yielding a natural operation $\Omega_E:\nabla E\rightarrow V\otimes\wedge^2(T\times X)$.  We can think of this as assigning ``curvatures'' to ``solutions'' of $E$, i.e., measuring, for each hyperlocal solution (to second order), its failure to be integrable to a solution over a neighborhood.

Then we have the following.
\begin{defn}
A natural equation $(E,e)$ on a natural bundle $X\rightarrow M$ has \emph{integrable constraints} if the natural constraint bundle $C$ associated with $(E,e)$ is isomorphic to $\ker(\Omega_E^*)$, the cokernel of $\Omega_E$.
\end{defn}
This definition says that a natural system of equations is integrable just in case its curvature at each point annihilates the presystem.

\section{Well-posedness and naturality}

Taken together, the foregoing almost establishes the following:
\begin{thm}\label{thm:existence}
    Let $(E,e)$ be a minimally coupled quasilinear naturally symmetric hyperbolic natural equation on a natural bundle $X=Y\times G$, admitting a natural constraint bundle, whose constraints that are integrable and complete relative to all timelike covectors.   Then given any manifold $M$, any section $\psi:M\rightarrow GM$, and any pair $(S,\varphi_0)$, where $S\subseteq M$ is an $n-1$ dimensional hypersurface of $M$ whose normal covector $n_a$ is everywhere timelike and $\varphi_0:S\rightarrow XM_{|\psi}$ satisfies all constraint equations, there exists a neighborhood $O$ of $S$ and a section $\varphi:O\rightarrow XM_{|\psi}$ that solves $E$ and agrees with $\varphi_0$ on $S$.
\end{thm}
This is the existence half of the classic existence and uniqueness theorems.  (Theorem \ref{thm:uniqueness} stated the other half.) We say ``almost'' establishes because, to our knowledge, the gap described by \citet[\S A]{Geroch2004} remains open.\footnote{Though, see \citet{Reula} and \citet{Abalos}.}

Note that the naturality of everything in this result establishes the following:
\begin{thm}
     Let $(E,e)$ be a minimally coupled quasilinear naturally symmetric hyperbolic natural equation on a natural bundle $X=Y\times G$, admitting a natural constraint bundle, whose constraints that are integrable and complete relative to all timelike covectors.  Given any $n-$manifold $M$, section $\psi:M\rightarrow GM$, local solution $\varphi:O\rightarrow EM_{|\psi}$, and diffeomorphism $\chi:O\rightarrow M'$, $\chi_*(\varphi):\chi[O]\rightarrow EM_{|\chi_*(\psi)}$ is also a solution.
\end{thm}
It follows that existence and uniqueness also hold ``naturally'' for suitably transformed initial data.  This captures the sense in which minimally coupled natural equations may be symmetric hyperbolic, up to a choice of (natural) background structure.

\section{Conclusion}

We have considered the relationship between symmetric hyperbolicity, natural equations, and minimal coupling, and we have shown how standard definitions of symmetric hyperbolicity may be adapted to a natural context.  We have also shown how minimal coupling imposes constraints on the causal structure of naturally symmetric hyperbolic equations, and how it implies a kind of conservation condition for stress-energy tensors suitably related to the hyperbolizability of an equation.  Finally, we have discussed how to understand constraints, including integrability, in this context.

\section*{Acknowldgements}
This material is based upon work supported by the National Science Foundation under Grant No. 2419967.
%\nocite{*}
\bibliography{covariance}% Produces the bibliography via BibTeX.

\end{document}